\documentclass{llncs}
 
\usepackage{microtype}
\usepackage{multibib}
\newcites{x}{References used in Appendix}
\usepackage{amsmath,amssymb}
\usepackage{tabls}
\usepackage[dvipdfmx]{graphicx}
\usepackage{wrapfig}
\usepackage{array}
\usepackage[algoruled,linesnumbered,algo2e,vlined]{algorithm2e}
\usepackage{url}
\usepackage{amssymb}
\usepackage{multirow}
\usepackage{multicol}

\usepackage[draft,mode=multiuser]{fixme}
\fxuseenvlayout{colorsig}
\fxusetargetlayout{color}
\FXRegisterAuthor{tn}{atn}{TN}
\FXRegisterAuthor{ti}{ati}{TI}
\FXRegisterAuthor{si}{asi}{SI}
\FXRegisterAuthor{hb}{ahb}{HB}
\FXRegisterAuthor{mt}{amt}{MT}
\definecolor{kgnote}{rgb}{1.0000,0.0000,0.0000}
\fxsetface{env}{}

\newtheorem{observation}{Observation}
\newtheorem{fact}{Fact}

\newcommand{\Occ}{\mathit{Occ}}

\newcommand{\vOcc}{\mathit{vOcc}}
\newcommand{\derive}{\mathit{val}}

\newlength\savedwidth

\newcommand{\rev}[1]{#1^{R}}

\newcommand{\LCPQ}{\mathsf{LCP}}
\newcommand{\LCSQ}{\mathsf{LCS}}
\newcommand{\LCEQ}{\mathsf{LCE}}

\newcommand{\shrink}[2]{\mathit{Shrink}_{#1}^{#2}}
\newcommand{\xshrink}[2]{\mathit{XShrink}_{#1}^{#2}}
\newcommand{\pow}[2]{\mathit{Pow}_{#1}^{#2}}
\newcommand{\xpow}[2]{\mathit{XPow}_{#1}^{#2}}
\newcommand{\uniq}[1]{\mathit{Uniq}(#1)}

\newcommand{\id}[1]{\mathit{id}(#1)}

\newcommand{\encblock}[1]{\mathit{Eblock}(#1)}

\newcommand{\encblockd}[1]{\mathit{Eblock}_{d}(#1)}
\newcommand{\encpow}[1]{\mathit{Epow}(#1)}

\newcommand{\deltaLR}[1]{\Delta_{#1}}

\newcommand{\sig}[1]{\mathit{Sig}(#1)}

\newcommand{\val}[1]{\mathit{val}(#1)}
\newcommand{\valp}[1]{\mathit{val}^{+}(#1)}

\newcommand{\lcp}[2]{\mathsf{LCP}(#1,#2)}
\newcommand{\lcs}[2]{\mathsf{LCS}(#1,#2)}

\title{Dynamic index, LZ factorization, and LCE queries in compressed space}

\author{
Takaaki Nishimoto\inst{1}
\and
Tomohiro I\inst{2}
\and
Shunsuke Inenaga\inst{1}
\and 
Hideo Bannai\inst{1}
\and \\
Masayuki Takeda\inst{1}
}
 \institute{
   Department of Informatics, Kyushu University, Japan \\
   \email{\{takaaki.nishimoto, inenaga, bannai, takeda\}@inf.kyushu-u.ac.jp}
   \and
   Kyushu Institute of Technology, Japan \\
   \email {tomohiro@ai.kyutech.ac.jp}
 }

\date{}

\begin{document}

\maketitle

\begin{abstract}
In this paper, we present the following results:
(1)
We propose a new \emph{dynamic compressed index} of $O(w)$ space, that
supports searching for a pattern $P$ in the current text in
$O(|P| f(M,w) + \log w \log |P| \log^* M (\log N + \log |P| \log^* M) + \mathit{occ} \log N)$ time and
insertion/deletion of a substring of length $y$ in
$O((y+ \log N\log^* M)\log w \log N \log^* M)$ time, where
$N$ is the length of the current text, $M$ is the maximum length of the dynamic text,
$z$ is the size of the Lempel-Ziv77 (LZ77) factorization of the
current text, $f(a,b) = O(\min \{ \frac{\log\log a \log\log b}{\log\log\log a}, \sqrt{\frac{\log b}{\log\log b}} \})$ 
and $w = O(z \log N \log^*M)$.
(2) We propose a new space-efficient LZ77 factorization algorithm for
a given text of length $N$, which runs in $O(N f(N,w') + z \log w'
\log^3 N (\log^* N)^2)$ time with $O(w')$ working space, where $w' =O(z \log N \log^* N)$.
(3) We propose a data structure of $O(w)$ space which supports longest
common extension (LCE) queries on the text in $O(\log N + \log \ell \log^* N)$
time, where $\ell$ is the output LCE length.
%
On top of the above contributions, we show several applications of our
data structures which improve previous best known results
on grammar-compressed string processing.
\end{abstract}

\section{Introduction}

\subsection{Dynamic compressed index}

In this paper, we consider the \emph{dynamic compressed text indexing problem}
of maintaining a compressed index for a text string that can be modified.
Although there exits several dynamic \emph{non-compressed} text indexes
(see e.g. ~\cite{DBLP:conf/focs/SahinalpV96,DBLP:conf/soda/AlstrupBR00} for recent work), there has been little work for the compressed variants.
Hon et al.~\cite{DBLP:conf/dcc/HonLSSY04} proposed 
the first dynamic compressed index of
$O(\frac{1}{\epsilon}(NH_0+N))$ bits of
space which supports searching of $P$ in $O(|P| \log^{2}
N(\log^{\epsilon} N + \log |\Sigma|) + \mathit{occ} \log^{1+\epsilon} N)$ time
and insertion/deletion of a substring of length $y$ in
$O((y+\sqrt{N})\log^{2+\epsilon} N)$ amortized time, where $0 <
\epsilon \leq 1$ and $H_0 \leq \log |\Sigma|$ denotes the zeroth order empirical entropy of the text of length $N$~\cite{DBLP:conf/dcc/HonLSSY04}.
Salson et al.~\cite{DBLP:journals/jda/SalsonLLM10} also proposed a dynamic compressed index, called \emph{dynamic FM-Index}. Although their approach works well in practice,
updates require $O(N \log N)$ time in the worst case.
To our knowledge, these are the only existing dynamic compressed indexes to date.

In this paper, we propose a new dynamic compressed index, as follows:
\begin{theorem}\label{theo:dynamic_index}
Let $M$ be the maximum length of the dynamic text to index,
$N$ the length of the current text $T$,
and $z$ the number of factors in the Lempel-Ziv 77 factorization of $T$
without self-references.
Then, there exist a dynamic index of $O(w)$ space 
which supports searching of a pattern $P$ 
in 
$O(|P| f_{\mathcal{A}} + \log w \log |P| \log^* M (\log N + \log |P| \log^* M) + \mathit{occ} \log N)$
time and
insertion/deletion of a substring of length $y$ in amortized $O((y+ \log N
\log^* M)\log w \log N \log^* M)$ time, where $w = O(z \log N \log^* M)$ and 
$f_{\mathcal{A}} = O(\min \{ \frac{\log\log M \log\log w}{\log\log\log M}, \sqrt{\frac{\log w}{\log\log w}} \})$.
\end{theorem}
Since $z \geq \log N$,
$\log w = \max\{\log z, \log(\log^* M)\}$.
Hence, our index is able to find pattern occurrences faster than the index of Hon et al.
when the $|P|$ term is dominating in the pattern search times.
Also, our index allows faster substring insertion/deletion on the text
when the $\sqrt{N}$ term is dominating.

\subsubsection{Related work.}\label{sec:related_work}
  Our dynamic compressed index uses Mehlhorn et al.'s
  \emph{locally consistent parsing} and \emph{signature encodings}
  of strings~\cite{DBLP:journals/algorithmica/MehlhornSU97},
  originally proposed for efficient equality testing of dynamic strings.
  Alstrup et al.~\cite{DBLP:conf/soda/AlstrupBR00} showed how to improve
  the construction time of Mehlhorn et al.'s data structure
  (details can be found in the technical report~\cite{LongAlstrup}).
  Our data structure uses Alstrup et al.'s
  fast string concatenation/split algorithms and 
  linear-time computation of locally consistent parsing,
  but has little else in common than those.
  In particular, Alstrup et al.'s dynamic pattern matching
  algorithm~\cite{DBLP:conf/soda/AlstrupBR00,LongAlstrup}
  requires to maintain specific locations called \emph{anchors}
  over the parse trees of the signature encodings,
  but our index does not use anchors.

Our index has close relationship to the ESP-indices~\cite{DBLP:conf/wea/TakabatakeTS14,TakabatakeTS15},
but there are two significant differences between ours and ESP-indices:
The first difference is that the ESP-index~\cite{DBLP:conf/wea/TakabatakeTS14}
is static and its online variant~\cite{TakabatakeTS15} allows only for
appending new characters to the end of the text,
while our index is fully dynamic allowing for insertion and deletion
of arbitrary substrings at arbitrary positions.
The second difference is that the pattern search time of the ESP-index 
is proportional to the number $\mathit{occ}_c$ of occurrences of the so-called ``core'' 
of a query pattern $P$, which corresponds to a maximal subtree of the 
ESP derivation tree of a query pattern $P$.
If $\mathit{occ}$ is the number of occurrences of $P$ in the text,
then it always holds that $\mathit{occ}_c \geq \mathit{occ}$,
  and in general $\mathit{occ}_c$ cannot be upper bounded by
  any function of $\mathit{occ}$.
In contrast, as can be seen in Theorem~\ref{theo:dynamic_index},
the pattern search time of our index is proportional to the
number $\mathit{occ}$ of occurrences of a query pattern $P$.
This became possible due to our discovery of 
a new property of the signature encoding~\cite{LongAlstrup}
(stated in Lemma \ref{lem:pattern_occurrence_lemma1}).
In relation to our problem, 
there exists the library management problem of maintaining 
a text collection (a set of text strings) allowing for insertion/deletion of texts 
(see \cite{DBLP:journals/corr/MunroNV15} for recent work).
While in our problem a single text is edited by insertion/deletion of substrings,
in the library management problem a text can be inserted to or deleted from the collection.
Hence, algorithms for the library management problem cannot be directly
applied to our problem.

\subsection{Applications and extensions}

\subsubsection{Computing LZ77 factorization in compressed space.}
  As an application to our dynamic compressed index,
  we present a new LZ77 factorization algorithm for a string $T$ of length $N$,
  running in $O(Nf_\mathcal{A} + z \log w \log ^3 N (\log^* N)^2)$
  time and $O(w)$ working space, where $f_{\mathcal{A}} = O(\min \{ \frac{\log\log N \log\log w}{\log\log\log N}, \sqrt{\frac{\log w}{\log\log w}} \})$. 
  Goto et al.~\cite{DBLP:journals/corr/abs-1107-2729} showed how,
  given the grammar-like representation for string $T$ generated
  by the LCA algorithm~\cite{DBLP:journals/ieicet/SakamotoMKS09},
  to compute the LZ77 factorization of $T$
  in $O(z \log^2 m \log^3 N + m \log m \log^3 N)$ time and $O(m \log^2 m)$ space, 
  where $m$ is the size of the given representation.
  Sakamoto et al.~\cite{DBLP:journals/ieicet/SakamotoMKS09} claimed that
  $m = O(z \log N \log^* N)$, however, it seems that in this bound they do not consider 
  the production rules to represent maximal runs of non-terminals in the derivation tree.
  The bound we were able to obtain with the best of our knowledge and 
  understanding is $m = O(z \log^2 N \log^* N)$,
  and hence our algorithm seems to use less space than the algorithm of Goto et al.~\cite{DBLP:journals/corr/abs-1107-2729}.
  Recently, Fischer et al.~\cite{FGGK15} showed
  a Monte-Carlo randomized algorithms to compute  
  an approximation of the LZ77 factorization
  with at most $2z$ factors in $O(N \log N)$ time,
  and another approximation with at most $(i+\epsilon)z$ factors
  in $O(N \log^2 N)$ time for any constant $\epsilon > 0$,
  using $O(z)$ space each.
  Another line of research is a recent result by
  Policriti and Prezza~\cite{PolicritiP15} which uses
  $N H_0 + o(N \log |\Sigma|) + O(|\Sigma| \log N)$ bits of space
  and computes the LZ77 factorization in $O(N \log N)$ time.

\subsubsection{Longest common extension queries in compressed space.}
Furthermore, we consider
the \emph{longest common extension} (LCE) problems on: 
an uncompressed string $T$ of length $N$; 
a grammar-compressed string $T$ represented by an straight-line program (SLP) of size $n$,
or an LZ77-compressed string $T$ with $z$ factors.
The best known deterministic LCE data structure on SLPs is
due to I et al.~\cite{IMSIBTNS15},
which supports LCE queries in $O(h \log N)$ time each,
occupies $O(n^2)$ space, and can be built in $O(h n^2)$ time,
where $h$ is the height of the derivation tree of a given SLP.
Bille et al.~\cite{bille13:_finger_compr_strin} showed 
a Monte Carlo randomized data structure built on a given SLP of size $n$
which supports LCE queries in $O(\log N \log \ell)$ time each,
where $\ell$ is the output of the LCE query and $N$ is the length of the uncompressed text.
Their data structure requires only $O(n)$ space,
but requires $O(N)$ time to construct.
Very recently, Bille et al.~\cite{BilleCCG15} showed
a faster Monte Carlo randomized data structure of $O(n)$ space
which supports LCE queries in $O(\log N + \log^2 \ell)$ time each.
The preprocessing time of this new data structure is not given in~\cite{BilleCCG15}.

In this paper, we present a new, deterministic 
LCE data structure using compressed space, namely $O(w)$ space,
supporting LCE queries in $O(\log N + \log \ell \log^* N)$ time each.
We show how to construct this data structure in 
$O(N \log w)$ time given an uncompressed string of length $N$,
$O(n \log\log n \log N \log^* N)$ time given an SLP of size $n$,
and $O(z \log w \log N \log^* N)$ time given the LZ77 factorization of size $z$.
We remark that our new LCE data structure allows for
fastest deterministic LCE queries on SLPs, 
and even permits faster LCE queries than 
the randomized data structure of Bille et al.~\cite{BilleCCG15} 
when $\log^* N = o(\log \ell )$ which in many cases is true.

All proofs omitted due to lack of space can be found in the appendices.

\section{Preliminaries} \label{sec:preliminary}

\subsection{Strings}

Let $\Sigma$ be an ordered alphabet and 
$\$$ be the lexicographically largest character in $\Sigma$.
An element of $\Sigma^*$ is called a string.
For string $w = xyz$,
$x$ is called a prefix,
$y$ is called a substring,
and $z$ is called a suffix of $w$, respectively.
The length of string $w$ is denoted by $|w|$.
The empty string $\varepsilon$ is a string of length 0,
that is, $|\varepsilon| = 0$.
Let $\Sigma^+ = \Sigma^* - \{\varepsilon\}$.
For any $1 \leq i \leq |w|$, $w[i]$ denotes the $i$-th character of $w$.
For any $1 \leq i \leq j \leq |w|$,
$w[i..j]$ denotes the substring of $w$ that begins at position $i$
and ends at position $j$.
Let $w[i..] = w[i..|w|]$ and $w[..i] = w[1..i]$ for any $1 \leq i \leq |w|$.
For any string $w$, let $\rev{w}$ denote the reversed string of $w$,
that is, $\rev{w} = w[|w|] \cdots w[2]w[1]$. 
For any strings $w$ and $u$, 
let $\lcp{w}{u}$ (resp. $\lcs{w}{u}$) denote the length of 
the longest common prefix (resp. suffix) of $w$ and $u$.
Given two strings $s_1, s_2$ and two integers $i, j$, let $\LCEQ(s_1, s_2, i, j)$ denote a query which returns $\lcp{s_1[i..|s_1|]}{s_2[j..|s_2|]}$.


For any strings $p$ and $s$, let $\Occ(p,s)$
denote all occurrence positions of $p$ in $s$,
namely, $\Occ(p, s) = \{i \mid p = s[i..i+|p|-1], 1 \leq i \leq |s|-|p|+1\}$.


In this paper, we deal with a dynamic text, namely,
we allow for insertion/deletion of a substring to/from an arbitrary position of the text.
Let $M$ be the maximum length of the dynamic text to index.
Our model of computation is the unit-cost word RAM with machine word 
size of $\log_2 M$ bits,
and space complexities will be evaluated by the number of machine words.
Bit-oriented evaluation of space complexities can be obtained with 
a $\log_2 M$ multiplicative factor.

\subsubsection{Lempel-Ziv 77 factorization.}

We will use the \emph{Lempel-Ziv 77} factorization~\cite{LZ77} of a string
to bound the running time and the size of our data structure on the string.
It is a greedy factorization which scans
the string from left to right, and
recursively takes as a factor
the longest prefix of the remaining suffix with a previous occurrence.
Formally, it is defined as follows.
\begin{definition}[Lempel-Ziv77 Factorization~\cite{LZ77}]
The Lempel-Ziv77 (LZ77) factorization of a string $s$ without self-references is 
a sequence $f_1, \ldots, f_z$ of non-empty substrings of $s$ 
such that $s = f_1 \cdots f_z$,
$f_1 = s[1]$, 
and for $1 < i \leq z$, if the character $s[|f_1..f_{i-1}|+1]$ does not occur in $s[|f_1..f_{i-1}|]$, then $f_i = s[|f_1..f_{i-1}|+1]$, otherwise $f_i$ is the longest prefix of $f_i \cdots f_z$ which occurs in $f_1 \cdots f_{i-1}$.
\end{definition}
The size of the LZ77 factorization $f_1, \ldots, f_z$ of string $s$
is the number $z$ of factors in the factorization.

A variant of LZ77 factorization which allows for self-overlapping reference
to a previous occurrence is formally defined as follows.
\begin{definition}[Lempel-Ziv77 Factorization with self-reference~\cite{LZ77}]
The Lempel-Ziv77 (LZ77) factorization of a string $s$ with self-references is
a sequence $f_1, \ldots, f_k$ of non-empty substrings of $s$ 
such that $s = f_1 \cdots f_k$,
$f_1 = s[1]$, 
and for $1 < i \leq k$, if the character $s[|f_1..f_{i-1}|+1]$ does not occur in $s[|f_1..f_{i-1}|]$, then $f_i = s[|f_1..f_{i-1}|+1]$, otherwise $f_i$ is the longest prefix of $f_i \cdots f_k$ which occurs at some position $p$, where $1 \leq p \leq |f_1 \cdots f_{i-1}|$.
\end{definition}
We will show that using our data structure,
the LZ77 with self-reference can be computed efficiently in compressed space.

\subsubsection{Locally consistent parsing.}

Let $p$ be a string of length $n$ over an integer alphabet of size $W$
where any adjacent elements are different, i.e., $p[i] \neq p[i+1]$ for
all $1 \leq i < n$.
A locally consistent parsing~\cite{DBLP:journals/algorithmica/MehlhornSU97}
of $p$ is a parsing (or factorization) $q_1, \ldots, q_j$ of $p$ such that
$p = q_1 \cdots q_j$, $2 \leq |q_h| \leq 4$ for any $1 \leq h \leq j$,
and the boundary between $q_{h-1}$ and $q_{h}$ is ``determined'' by
$p[|q_1 \cdots q_{h-1}| + 1 - \Delta_L..|q_1 \cdots q_{h-1}| + 1 + \Delta_R]$,
where $\Delta_L = \log^*W + 6$ and $\Delta_R = 4$.
Clearly, $j \leq n/2$.
By ``determined'' above, we mean that
if a position $i$ of an integer string $p$
and a position $k$ of another integer string $s$ share the
same left context of length at least $\Delta_L$ and
the same right context of length at least $\Delta_R$,
then there is a boundary of the locally consistent parsing of $p$ 
between the positions $i-1$ and $i$ iff there is a boundary of the locally consistent
parsing of $s$ between the positions $k-1$ and $k$.
A formal definition of locally consistent parsing,
and its linear-time computation algorithm, is explained in the following lemma.

\begin{lemma}[Locally consistent parsing~\cite{DBLP:journals/algorithmica/MehlhornSU97,LongAlstrup}]\label{lem:CoinTossing}
  Let $W$ be a non-negative integer and
  let $p$ be an integer sequence of length $n$,
  called a \emph{$W$-colored sequence},
  where $p[i] \neq p[i+1]$ for any $1 \leq i < n$ and $0 \leq p[j] \leq W$ for any $1 \leq j \leq n$.
  For every $W$ there exists a function $f : [-1..W]^{\log^* W + 11} \rightarrow \{0,1\}$ 
  such that for every $W$-colored sequence $p$, 
  the bit sequence $d$ defined by 
  $d[i] = f(\tilde{p}[i-\deltaLR{L}], \ldots, \tilde{p}[i+\deltaLR{R}])$ for $1 \leq i \leq n$, satisfies:
\begin{itemize}
 \item $d[1]=1$,
 \item $d[i] + d[i+1] \leq 1$ for $1 \leq i < n$, and
 \item $d[i] + d[i+1] + d[i+2] + d[i+3] \geq 1$ for any $1 \leq i < n-3$,
\end{itemize}
where $\deltaLR{L} = \log ^*W + 6$, $\deltaLR{R} = 4$, and $\tilde{p}[j] = p[j]$ for all $1 \leq j \leq n$, $\tilde{p}[j] = -1$ otherwise. 
Furthermore, $d$ can be computed in $O(n)$ time using a precomputed table of size $o(\log W)$. 
Also, we can compute this table in $o(\log W)$ time.
\end{lemma}
\begin{proof}
Here we give only an intuitive description of a proof of Lemma~\ref{lem:CoinTossing}.
More detailed proofs can be found at~\cite{DBLP:journals/algorithmica/MehlhornSU97}
and~\cite{LongAlstrup}.

Mehlhorn et al.~\cite{DBLP:journals/algorithmica/MehlhornSU97} showed that there exists a function $f'$  
which returns a $(\log W)$-colored sequence $p'$ for a given $W$-colored sequence $p$ in $O(|p|)$ time, 
where $p'[i]$ is determined only by $p[i-1]$ and $p[i]$ for $1 \leq i \leq |p|$. 
Let $p^{\langle k \rangle}$ denote the outputs after applying $f'$ to $p$ by $k$ times.
They also showed that there exists a function $f''$
which returns a bit sequence $d$ satisfying the conditions of Lemma~\ref{lem:CoinTossing} 
for a $6$-colored sequence $p$ in $O(|p|)$ time, 
where $d[i]$ is determined only by $p[i-3..i+3]$ for $1 \leq i \leq |p|$. 
Hence we can compute $d$ for a $W$-colored sequence $p$ 
in $O(|p| \log^* W)$ time by applying $f''$ to $p^{\langle \log^* W + 2 \rangle}$ after computing $p^{\langle \log^* W + 2 \rangle}$. 
Furthermore, Alstrup et al.~\cite{LongAlstrup} showed
that $d$ can be computed in $O(|p|)$ time using a precomputed table of size $o(\log W)$.
The idea is that $p^{\langle 3 \rangle}$ is a $\log\log\log W$-colored sequence and the number of 
all combinations of a $\log\log\log W$-colored sequence of length $\log^* W + 11$ is $2^{(\log^* W + 11)\log\log\log W} = o(\log W)$. 
Hence we can compute $d$ for a $W$-colored sequence in linear time using a precomputed table of size $o(\log W)$.
\qed
\end{proof}

Given a bit sequence $d$ of Lemma~\ref{lem:CoinTossing},
let $\mathit{Eblock}_d(p)$ be the function that 
decomposes an integer sequence $p$ into a sequence
$q_1, \ldots, q_j$ of substrings called \emph{blocks} of $p$,
such that $p = q_1 \cdots q_j$ and 
$q_i$ is in the decomposition iff $d[|q_1 \cdots q_{i-1}|+1] = 1$
for any $1 \leq i \leq j$.
We omit $d$ and write $\encblockd{p}$ when it is clear from the context, 
and we use implicitly the bit sequence created by Lemma~\ref{lem:CoinTossing} as $d$. 
Let $|\encblock{p}| = j$ and let $\encblock{s}[i] = q_i$.
For a string $s$, 
let $\encpow{s}$ be the function which 
groups each maximal run of same characters $a$ as $a^r$,
where $r$ is the length of the run. 
$\encpow{s}$ can be computed in $O(|s|)$ time.
Let $|\encpow{s}|$ denote the number of maximal runs of same characters in $s$ and let $\encpow{s}[i]$ denote $i$-th maximal run in $s$.
\begin{example}[$\encblockd{p}$ and $\encpow{s}$] \label{exmp:Encblock}
Let $\log^* W = 2$, and then $\deltaLR{L} = 8, \deltaLR{R} = 4$.\\
If $p = 1,2,3,2,5,7,6,4,3,4,3,4,1,2,3,4,5$ and $d = 1,0,0,1,0,1,0,0,1,0,0,0,1,0,1,0,0$,
then $\encblockd{p} = (1,2,3),(2,5),(7,6,4),(3,4,3,4),(1,2),(3,4,5)$, $|\encblockd{p}| = 6$ and $\encblockd{p}[2] = (2, 5)$. 
For string $s = aabbbbbabb$,
$\encpow{s} = a^2b^5a^1b^2$ and
$|\encpow{s}| = 4$ and $\encpow{s}[2] = b^5$.
\end{example}

\subsection{Context free grammars as compressed representation of strings}

\subsubsection{Admissible context free grammars.}
An \emph{admissible context free grammar} (ACFG)~\cite{KiefferY00}
is a CFG which generates only a single string.
More formally, 
an ACFG that generates a single string $T$ is
a quadruple $\mathcal{G} = (\Sigma, \mathcal{V}, \mathcal{D}, S)$, 
such that
\begin{itemize}
 \item $\Sigma$ is an ordered alphabet of terminal characters,
 \item $\mathcal{V} = \{ e_1, \ldots, e_{k} \}$ is a set of positive integers
       with $e_1 < \cdots < e_k$, called \emph{variables},
 \item $\mathcal{D} = \{e_i \rightarrow \mathit{xexpr}_i\}_{i = 1}^{k}$ is
       a set of \emph{deterministic productions} (or \emph{assignments})
       i.e., for each variable $e \in \mathcal{V}$ there is exactly one production in $\mathcal{D}$
       whose lefthand side is $e$,
 \item each $e_i \in \mathcal{V} \setminus \{S\}$ appears at least once in the righthand side of some production $e_j \rightarrow \mathit{xexpr}_j$ with $e_i < e_j$, and
 \item $S \in \mathcal{V}$ is the start symbol which derives the string $T$.
\end{itemize}
Sometimes we handle a variable sequence as a kind of string. 
For example, for any variable sequence $y = e_{i_1} \cdots e_{i_d} \in \mathcal{V}^{+}$, 
let $|y| = d$ and $y[c] = e_{i_c}$ for $1 \leq c \leq d$.
Let $\mathit{val}: \mathcal{V} \rightarrow \Sigma^+$ be the function
which returns the string derived by an input variable.
If $s = \val{e}$ for $e \in \mathcal{V}$,
then we say that the variable $e$ \emph{represents} string $s$.
For any variable sequence $y \in \mathcal{V}^{+}$,
let $\valp{y} = \val{y[1]} \cdots \val{y[|y|]}$.

%

For two variables $e_1, e_2 \in \mathcal{V}$, we say that $e_1$ occurs at position $c$ in $e_2$ if there is a node labeled with $e_1$ in the derivation tree of $e_2$
and the leftmost leaf of the subtree rooted at that node labeled with $e_1$
is the $c$-th leaf in the derivation tree of $e_2$. 
Furthermore, for variable sequence $y \in \mathcal{V}^{+}$, 
we say that $y$ occurs at position $c$ in $e$ if $y[i]$ occurs at position $c+|\valp{y[..i-1]}|$ in $e$ for $1 \leq i \leq |y|$. 
We define the function $\vOcc(e_1,e_2)$ which returns all positions of $e_1$ in the derivation tree of $e_2$. 

\subsubsection{Straight-line programs.}
A \emph{straight-line program} (\emph{SLP}) is an ACFG in the Chomsky normal from.
Formally, SLP $\mathcal{S}$ of size $n$ is an ACFG
$\mathcal{G} = (\Sigma, \mathcal{V}, \mathcal{D}, X_n)$, 
where $\mathcal{V} = \{ X_{1}, \cdots, X_n \}$, $\val{X_n} = T$, 
$\mathcal{D} = \{X_i \rightarrow \mathit{expr}_i\}_{i = 1}^{n}$
with each $\mathit{expr}_i$ being either of form
$X_\ell X_r~(1 \leq \ell, r < i)$,
or a single character $a \in \Sigma$.
The size of the SLP $\mathcal{G}$ is the number $n$ of productions in $\mathcal{D}$.
In the extreme cases the length $N$ of the string $T$ can be as large as $2^{n-1}$,
however, it is always the case that $n \geq \log_2 N$.
For any variable $X_i$ with $X_i \rightarrow X_{\ell}X_{r} \in \mathcal{D}$, 
let $X_i.{\rm left} = \val{X_{\ell}}$ and $X_i.{\rm right} = \val{X_{r}}$,
which are called 
the \emph{left string} and the \emph{right string} of $X_i$, respectively.
%
\begin{example}[SLP]\label{ex:SLP}
  Let $\mathcal{S} = (\Sigma, \mathcal{V}, \mathcal{D}, S)$ be the SLP
  s.t.
$\Sigma = \{A, B, C \}$, $\mathcal{V} = \{ X_1, \cdots , X_{11} \}$, 
$\mathcal{D} = \{ 
X_{1} \rightarrow A, X_{2} \rightarrow B, X_{3} \rightarrow C, 
X_4 \rightarrow X_{3}X_{1}, X_5 \rightarrow X_{4}X_{2}, 
X_6 \rightarrow X_{5}X_{5}, X_7 \rightarrow X_{2}X_{3}, 
X_8 \rightarrow X_{1}X_{2}, X_9 \rightarrow X_{7}X_{8}, 
X_{10} \rightarrow X_{6}X_{9}, X_{11} \rightarrow X_{10}X_{6}
\}$, $S = X_{11}$, 
the derivation tree of $S$ represents $CABCABBCABCABCAB$.
\end{example}

\subsubsection{Run-length ACFGs.}

We define \emph{run-length ACFGs} as an extension to ACFGs,
which allow \emph{run-length encodings} in the righthand sides of productions.
Formally, a run-length ACFG 
is $\mathcal{G} = (\Sigma, \mathcal{V}, \mathcal{D}, S)$, 
where $\mathcal{D} = \{e_i \rightarrow \mathit{xexpr}_i\}_{i = 1}^{w}$,
$\val{S} = T$ and
each $\mathit{xexpr}_i$ is in one of the following forms:
\begin{eqnarray}
 \mathit{xexpr}_i &=&
  \begin{cases}
   a \in \Sigma, \\
   e_{\ell}e_r \in \mathcal{V}^+ & (e_{\ell},e_r < e_i), \\
   \hat{e}^{d} \in \mathcal{V} \times \mathcal{N} & (\hat{e} < e_i, \mbox{ and } d > 1). \\
  \end{cases} \nonumber
\end{eqnarray}
Hence $\mathit{xexpr}_i \in \Sigma \cup \mathcal{V}^{+} \cup (\mathcal{V} \times \mathcal{N})$.
The \emph{size} of the run-length ACFG $\mathcal{G}$ is the number $w$ of productions in $\mathcal{D}$.

Let $\mathit{Sig}_{\mathcal{G}}:\Sigma \cup \mathcal{V}^{+} \cup (\mathcal{V} \times \mathcal{N}) \rightarrow \mathcal{V}$ be the function such that
 \begin{eqnarray}
 \mathit{Sig}_{\mathcal{G}}(\mathit{x}) &=&
  \begin{cases}
   e & \mbox{if } (e \rightarrow \mathit{x}) \in \mathcal{D}, \\
   \mathit{Sig}_{\mathcal{G}}( \mathit{Sig}_{\mathcal{G}}(x[1..|x|-1]) x[|x|])
   & \mbox{if } x[i] \in \mathcal{V} \mbox{ for } 1 \leq i \leq |x|, 2 < |x| \leq 4, \\
   \mbox{undefined } & \mbox{otherwise.}
  \end{cases} \nonumber
  \end{eqnarray}
 Namely, the function $\mathit{Sig}_{\mathcal{G}}$ returns, if any,
 the lefthand side of the corresponding production for
 a given element in $\Sigma \cup \mathcal{V}^{+} \cup (\mathcal{V} \times \mathcal{N})$
 of length 3 or 4, by recursively applying the $\mathit{Sig}$ function from left to right. 
Let $\mathit{Assgn}_{\mathcal{G}}$ be the function such that
$\mathit{Assgn}_{\mathcal{G}}(e_i) = \mathit{xexpr_i}$ iff $e_i \rightarrow \mathit{xexpr_i} \in \mathcal{D}$.
When clear from the context, 
we write $\mathit{Sig}_{\mathcal{G}}$ and $\mathit{Assgn}_{\mathcal{G}}$
as $\mathit{Sig}$ and $\mathit{Assgn}$, respectively.
  For any $p \in (\Sigma \cup \mathcal{V}^{+} \cup (\mathcal{V} \times \mathcal{N}))^*$,
  let $\mathit{Sig}^{+}(p) = \sig{p[1]} \cdots \sig{p[|p|]}$.
We define the left and right strings for any variable
$e_i \rightarrow e_{\ell}e_r \in \mathcal{D}$ in a similar way to SLPs. 
Furthermore, for any $e_i \rightarrow \hat{e}^{k} \in \mathcal{D}$,
let $e_i.{\rm left} = \val{\hat{e}}$ and $e_i.{\rm right} = \val{\hat{e}}^{k-1}$.

In this paper, we consider a DAG of size $w$ 
that is a compact representation of the derivation trees of variables in
a run-length ACFG $\mathcal{G}$,
where each node represents a variable in $\mathcal{V}$
and out-going edges represent the assignments in $\mathcal{D}$.
For example, if there exists an assignment $e_i \rightarrow e_{\ell}e_r \in \mathcal{D}$,
then there exist two out-going edges from $e_i$ to its ordered children
$e_{\ell}$ and $e_r$.
In addition, $e_{\ell}$ and $e_r$ have reversed edges to their parent $e_i$.
For any $e \in \mathcal{V}$,
let $\mathit{parents}(e)$ be the set of 
variables which have out-going edge to $e$ in the DAG of $\mathcal{G}$. 
If a node is labeled by $e$, then the node is associated with $|\val{e}|$.

\begin{example}[Run-length ACFG]\label{ex:tree}
Let $\mathcal{G} = (\Sigma, \mathcal{V}, \mathcal{D}, S)$ be a run-length ACFG, 
where $\Sigma = \{A, B, C \}$, $\mathcal{V} = \{1, \ldots , 15 \}$, $\mathcal{D} = \{ 
1 \rightarrow A, 2 \rightarrow B, 3 \rightarrow C,
4 \rightarrow 3^4, 5 \rightarrow 1^1, 6 \rightarrow 2^1, 7 \rightarrow 3^1, 
8 \rightarrow (7,5),  9 \rightarrow (8,6), 10 \rightarrow (5,6), 11 \rightarrow (10,4), 
12 \rightarrow 9^2, 13 \rightarrow 10^7, 14 \rightarrow 11^1, 15 \rightarrow (12,13), 16 \rightarrow (15,14), 
17 \rightarrow 16^1
\}$, and $S = 17$.
The derivation tree of the start symbol $S$ represents a single string 
$T = CABCABABABABABABABABABCCCC$. 
Here, $4.{\rm left} = \val{3}$, $4.{\rm right} = \valp{3^3}$, 
$\sig{(7,5)} = 8$, $\sig{(7,5,6)} = 9$, $\sig{(6,5)} = \rm{undefined}$, 
$\mathit{parents}(5) = \{8,10 \}$ and $\mathit{vOcc}(9,17) = \{1,4 \}$.
See also Fig.~\ref{fig:SignatureTree} in Section~\ref{sec:Framework}
which illustrates the derivation tree of the start symbol $S$
and the DAG for $\mathcal{G}$.
\end{example}


\subsubsection{Dynamization and data structure of run-length ACFG.}

In this paper, we consider a compressed representation and compressed index
of a dynamic text based on run-length ACFGs.
Hence, upon edits on the text,
the run-length ACFG representing the text needs to be modified as well.
To this end, we consider \emph{dynamic run-length ACFGs},
which allow for insertion of new assignments to $\mathcal{D}$,
and allow for deletion of assignments $e \rightarrow \mathit{xexpr}$
from $\mathcal{D}$ only if $|\mathit{parents}(e)| = 0$.
We remark that the grammar under modification may temporarily represents
more than one text,  
however, this will be readily fixed as soon as
we insert a new start symbol of the grammar
representing the edited text.

Next, we consider an abstract data structure 
$\mathcal{H}(f_{\mathcal{A}},f'_{\mathcal{A}})$ to maintain
a dynamic run-length ACFG 
$\mathcal{G} = (\Sigma, \mathcal{V}, \mathcal{D} , S)$ of size $w$. 
$\mathcal{H}$ consists of two components $\mathcal{A}$ and $\mathcal{B}$. 
The first component $\mathcal{A}$ is an abstract
data structure of $O(f'_{\mathcal{A}})$ size which
is able to add/remove an assignment to/from $\mathcal{D}$
in $O(f_{\mathcal{A}})$ time. 
This data structure is also able to compute $\sig{xexpr}$ in $O(f_{\mathcal{A}})$ time. 
For example, using a balanced binary search tree for $\mathcal{D}$,
we achieve deterministic $f_{\mathcal{A}} = O(\log w)$ time and 
$f'_{\mathcal{A}} = O(w)$ space.
Note that using the best known deterministic predecessor/successor
data structure for a dynamic set of integers~\cite{DBLP:journals/jcss/BeameF02},
we achieve deterministic
$f_{\mathcal{A}} = O\left(\min \left\{ \frac{\log \log M \log \log w}{\log \log \log M}, \sqrt{\frac{\log w}{\log\log w}} \right\} \right)$ time 
and $f'_{\mathcal{A}} = O(w)$ space, where $M$ is the maximum length of the dynamic text\footnote{
  Alstrup et al.~\cite{LongAlstrup} used hashing to maintain $\mathcal{A}$ and
  obtained a randomized $\mathcal{H}(1, w)$ signature dictionary. However,
  since we are interested in the worst case time complexities, we use
  balanced binary search trees or the data structure~\cite{DBLP:journals/jcss/BeameF02} in place of hashing.
}.
The second component $\mathcal{B}$ is the DAG of $\mathcal{G}$ introduced in the
previous subsection.
The corresponding nodes and edges of the DAG can be added/deleted
in constant time per addition/deletion of an assignment.
By maintaining $\mathit{parents}(e)$ with a doubly-linked list
for each node $v_e$ representing a variable $e \in \mathcal{V}$,
we obtain the following lemma:
\begin{lemma}\label{lem:basic_operation_time}
Using $\mathcal{H}(f_{\mathcal{A}}, f'_{\mathcal{A}})$ for 
a dynamic run-length ACFG $\mathcal{G} = (\Sigma, \mathcal{V}, \mathcal{D}, S)$
of current size $w$,
$\mathit{Sig}(\mathit{xexpr})$ can be computed in $O(f_{\mathcal{A}})$ time, 
for a given $\mathit{xexpr} \in \Sigma \cup {\mathcal{V}}^{+} \cup (\mathcal{V} \times \mathcal{N})$.
Given a node $v_e$ representing a variable $e$,
$\mathit{Assgn}(e)$ and $|\mathit{val}(e)|$ can be computed in $O(1)$ time, and 
$\mathit{parents}(e)$ can be computed in $O(|\mathit{parents}(e)|)$ time.
We can also update $\mathcal{H}$ in $O(f_{\mathcal{A}})$ time 
when an assignment is added to/removed from $\mathcal{D}$.
\end{lemma}
Note that $\mathit{Assgn}(e)$, $\mathit{Sig}(\mathit{xexpr})$ and 
$\mathit{parents}(e)$ can return not only the signatures but also the
corresponding nodes in the DAG.


\section{Signature encoding}\label{sec:Framework}

In this section, we recall the \emph{signature encoding} 
first proposed by Mehlhorn et al.~\cite{DBLP:journals/algorithmica/MehlhornSU97}. 
The signature encoding of a string $T$ is a run-length ACFG
$\mathcal{G} = (\Sigma, \mathcal{V}, \mathcal{D}, S)$
where the assignments in $\mathcal{D}$ are determined by
recursively applying to $T$
the locally consistent parsing, 
the $\mathit{Encblock}$ function, and the $\mathit{Sig}$ function 
(recall Section~\ref{sec:preliminary}),
until a single integer $S$ is obtained.
More formally, we use the $\mathit{Shrink}$ and $\mathit{Pow}$ functions
in the signature encoding of string $T$ defined below:

\begin{eqnarray*}
 \shrink{t}{T} &=&
  \begin{cases}
   \mathit{Sig}^{+}(T) & \mbox{for } t = 0 \\
   \mathit{Sig}^{+}(\encblock{\pow{t-1}{T}}) &  \mbox{for } 0 < t \leq h, \\
  \end{cases}\\
 \pow{t}{T} &=& \mathit{Sig}^{+}(\encpow{\shrink{t}{T}}) \ \ \ \ \ \ \mbox{for } 0 \leq t \leq h,
\end{eqnarray*}
where $h$ is the minimum integer satisfying $|\pow{h}{T}| = 1$.
  Then, the start symbol of the signature encoding is $S = \pow{h}{T}$,
  and the height of the derivation tree of the signature encoding of $T$ is
  $O(h) = O(\log N)$, where $N = |T|$
  (see also Fig.~\ref{fig:SignatureTree} below).

\begin{example}[Signature encoding]\label{ex:signature_dictionary}
Let $\mathcal{G} = (\Sigma, \mathcal{V}, \mathcal{D}, S)$ be a run-length ACFG of Example~\ref{ex:tree}. 
Assuming $\encblock{\pow{0}{T}} = (7,5,6),(7,5,6),(5,6)^7,(5,6,4)$ and $\encblock{\pow{1}{T}} = (12,13,14)$ hold, 
$\mathcal{G}$ is the signature encoding of $T$ and $\id{T} =  17$.
See Fig.~\ref{fig:SignatureTree} for an illustration of the derivation tree of $\mathcal{G}$ and the corresponding DAG.
\end{example}

\begin{figure}[ht]
\begin{center}
  \includegraphics[scale=0.7]{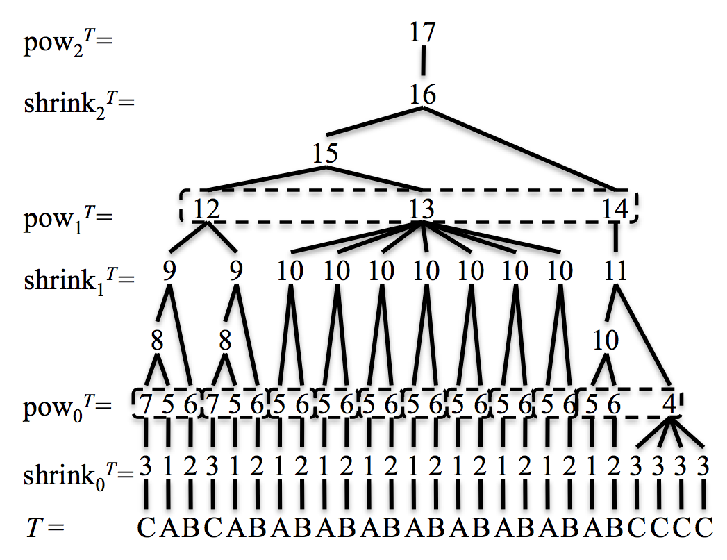}  
  \includegraphics[scale=0.7]{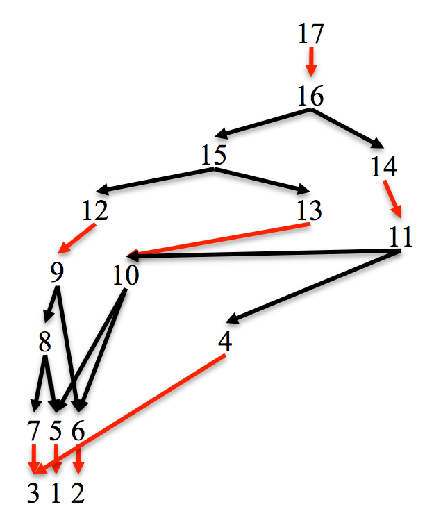}
  \caption{
  The derivation tree of $S$ (left) and the DAG for $\mathcal{G}$ (right) of Example~\ref{ex:tree}. 
  In the DAG, the black and red arrows represent $e \rightarrow e_le_r$ and
  $e \rightarrow \hat{e}^{k}$ respectively.   
  In Example~\ref{ex:signature_dictionary}, $T$ is encoded by signature encoding.
  In the derivation tree of $S$, the dotted boxes represent the blocks created by the $\mathit{Eblock}$ function.
  }
  \label{fig:SignatureTree}
\end{center}
\end{figure}

Each variable of the signature encoding (the run-length ACFG defined this way)
is called a \emph{signature}.
For any string $P \in \Sigma^+$,
let $\id{T} = \pow{h}{T} = S$, i.e.,
the integer $S$ is the signature of $T$.

The signature encoding of a text $T$ can be efficiently maintained
under insertion/deletion of arbitrary substrings to/from $T$.
For this purpose, we use the $\mathcal{H}(f_{\mathcal{A}}, f'_{\mathcal{A}})$
data structure for the signature encoding of a dynamic text.

\subsection{Properties of signature encodings}
Here we describe a number of useful properties of signature encodings.
The ones with references to the literature are known
but we provide their proofs for completeness.
The other ones without references are our new discoveries.

\subsubsection{Substring extraction.}
By the definition of the $\mathit{Eblock}$ function and Lemma~\ref{lem:CoinTossing},
for any $1 \leq t \leq h$,
$|\shrink{t}{T}| \leq |\pow{t-1}{T}|/2$ and $|\pow{t}{T}| \leq |\pow{t-1}{T}|/2$.
Thus $h \leq \log |s|$ and the height of the derivation tree of $e$ is $O(\log |\val{e}|)$ for any signature $e \in \mathcal{V}$.
Since each node of the DAG for a signature encoding
stores the length of the corresponding string, 
we have the following:
\begin{fact}\label{obs:RandomAccess}
  Using the DAG for a signature encoding $\mathcal{G} = (\Sigma, \mathcal{V}, \mathcal{D}, S)$ of size $w$,
  given a signature $e \in \mathcal{V}$ (and its corresponding node in the DAG),
  and two positive integers $i,k$, 
  we can compute $\val{e}[i..i+k-1]$ in $O(\log |\val{e}| + k)$ time. 
\end{fact}

\subsubsection{Space requirement of the signature encoding.}

Recall that we handle dynamic text of length at most $M$.
Then, the maximum value of the signatures is bounded by $3M-1$,
since the derivation tree can contain at most $M$ leaves,
and $2M-1$ internal nodes
(when there are no runs of same signatures at any height
of the derivation tree,
$\mathit{Pow}$ function generates
as many signatures as the $\mathit{Shrink}$ function).
We also remark that 
the input of the $\mathit{Eblock}$ function is a sequence of signatures.
Hence, $\deltaLR{L}$ of Lemma~\ref{lem:CoinTossing} is bounded by $\log^* 3M + 6 = O(\log^* M)$. 
Note that we can bound $M = \Theta(|T|)$ if we do not update $\mathcal{G}$ after we compute $\id{T}$.

Let $N$ be the length of the current text $T$.
The size $w$ of the signature encoding of $T$
is bounded by $3N-1$ by the same reasoning as above.
Also, the following lemma shows that
the signature encoding of $T$ requires only \emph{compressed space}:
\begin{lemma}[\cite{18045}]\label{lem:upperbound_signature}
The size $w$ of the signature encoding of $T$ is 
$O(z \log N \log^* M)$,
where $z$ is the number of factors in the LZ77 factorization without self-reference of $T$. 
\end{lemma}

\begin{proof}
See Appendix~\ref{sec:Proof_HConstructuionTheorem}.
\end{proof}

Hence, we have $w = O(\min\{z \log N \log^* M, N\})$.
In the sequel, we assume $z \log N \log^* M \leq N$
and will simply write $w = O(z \log N \log^* M)$,
since otherwise we can use some uncompressed dynamic text index
in the literature.

\subsubsection{Common sequences to all occurrences of same substrings.}
Here, we recall the most important property of the signature encoding.



Let $\mathcal{G} = (\Sigma, \mathcal{V}, \mathcal{D}, S)$
be the signature encoding of text $T$.
Let $i, j$~($i < j$) be any positions in $T$, and let $P = T[i..j]$.
Let $\mathbb{P}_i$ and $\mathbb{P}_{j}$ be the paths from
the root of the derivation tree of $\mathcal{G}$
to the $i$th and $j$th leaves, respectively.
Then, at each depth $\ell$ of the derivation tree of $\mathcal{G}$,
consider the sequence $s_{i, j, \ell}$ of signatures which lie to the right of
$\mathbb{P}_i$ with offset $\Delta_L + 3$ 
and to the left of $\mathbb{P}_j$ with offset $\Delta_R + 2$.
By the property of locally consistent parsing of Lemma~\ref{lem:CoinTossing},
$s_{i, j, \ell} = s_{i', j', \ell}$ for any occurrences $[i'..j']$ of $P$
in $T$ and for any depth $\ell$.
We call each signature contained in $s_{i, j, \ell}$
a \emph{consistent signature} w.r.t. $P$.

Formally, we define the consistent signatures of $P$ in the derivation tree
of $\mathcal{G}$ by the $\emph{XShrink}$ and $\emph{XPow}$ functions below,
where the prefix of length at least $\Delta_L$ and the suffix of length at least 
$\Delta_R+1$ are ``ignored'' at each depth of recursion:
\begin{definition}\label{def:xshrink}
For a string $P$, let
\begin{eqnarray*}
 \mathit{XShrink}_t^{P} &=&
  \begin{cases}
   \mathit{Sig}^{+}(P) & \mbox{ for } t = 0, \\
   \mathit{Sig}^{+}(\encblockd{\mathit{XPow}_{t-1}^{P}}[|L_{t}^{P}|..|\mathit{XPow}_{t-1}^{P}|-|R_{t}^{P}|]) & \mbox{ for } 0 < t \leq h^{P}, \\
  \end{cases} \\
\mathit{XPow}_t^{P} &=& \mathit{Sig}^{+}(\encpow{\mathit{XShrink}_t^{P}[|\hat{L}_{t}^{P}| + 1..|\mathit{XShrink}_t^{P}| - |\hat{R}_{t}^{P}]}|) \ \mbox{ for } 0 \leq t < h^{P}, 
\end{eqnarray*}
where
\begin{itemize}
  \item $L_{t}^{P}$ is the shortest prefix of $\mathit{XPow}_{t-1}^{P}$ of length at least $\deltaLR{L}$ such that $d[|L_{t}^{P}|+1] = 1$,
  \item $R_{t}^{P}$ is the shortest suffix of $\mathit{XPow}_{t-1}^{P}$ of length at least $\deltaLR{R}+1$ such that $d[|d| - |R_{t}^{P}| + 1] = 1$,
  \item $\hat{L}_{t}^{P}$ is the longest prefix of $\mathit{XShrink}_t^{P}$ such that $|\encpow{\hat{L}_{t}^{P}}|  = 1$,
  \item $\hat{R}_{t}^{P}$ is the longest suffix of $\mathit{XShrink}_t^{P}$ such that $|\encpow{\hat{R}_{t}^{P}}| = 1$, and
  \item and $h^{P}$ is the minimum integer such that 
$|\encpow{\mathit{XShrink}_{h^{P}}^{P}}| \leq \Delta_{L} + \Delta_{R} + 9$. 
\end{itemize}
\end{definition}
Note that $\Delta_{L} \leq |L_{t}^{P}| \leq \Delta_{L} + 3$ and $\Delta_{R}+1 \leq |R_{t}^{P}| \leq \Delta_{R} + 4$ hold by the definition. 
Hence $|\xshrink{t+1}{P}| > 0$ holds if $|\encpow{\xshrink{t}{P}}| > \Delta_{L} + \Delta_{R} + 9$. 
See Fig.~\ref{fig:CommonSequence} for illustrations of
consistent signatures of each occurrence of $P$ in $T$,
which are represented by the gray boxes.
Since at each depth we have ``ignored'' the left and right contexts
of respective length at most $\Delta_L+3$ and $\Delta_R+4$,
the consistent signatures at each depth are determined
only by the consistent signatures at the previous depth (1 level deeper).
This implies that for \emph{any} occurrences of $P$ in $T$,
there are common consistent signatures (gray boxes),
which will simply be called the \emph{common signatures} of $P$.
The next lemma formalizes this argument.


\begin{figure}[h]
\begin{center}
  \includegraphics[scale=0.5]{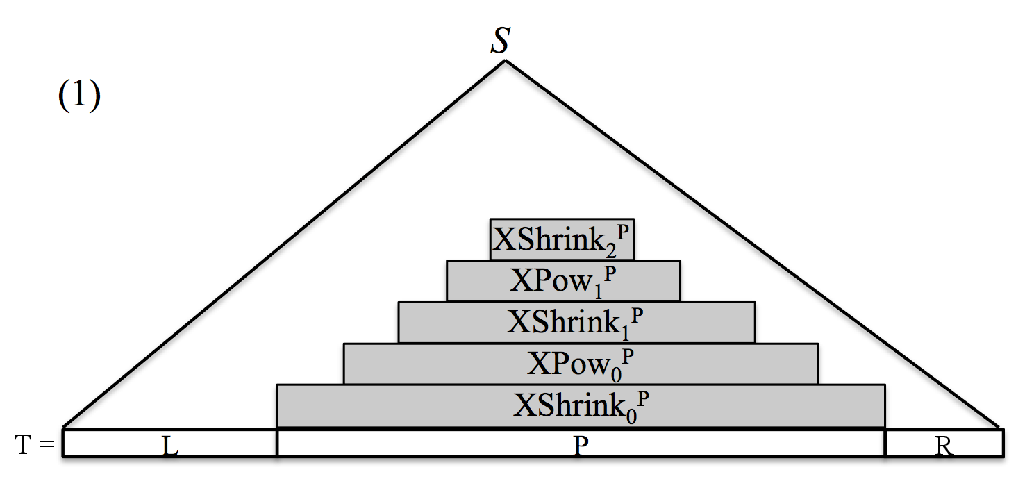}
  \includegraphics[scale=0.6]{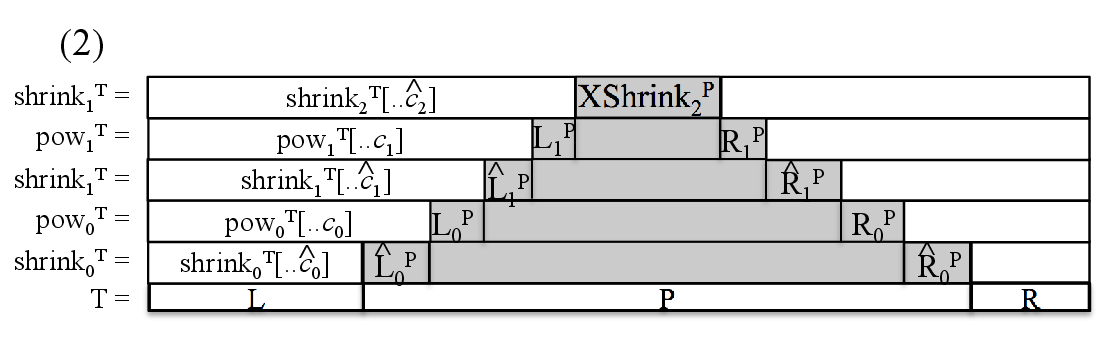}
  \includegraphics[scale=0.5]{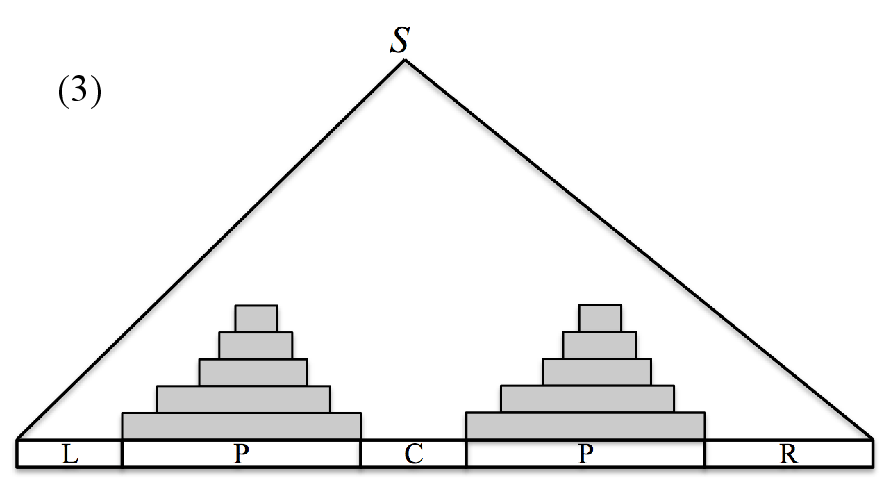}
  \includegraphics[scale=0.6]{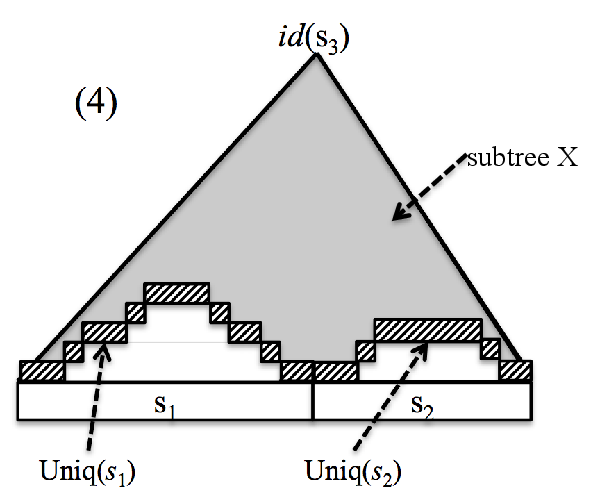}
  \caption{
  Abstract images of consistent signatures of substring $P$ of text $T$,
  on the derivation trees of the signature encoding of $T$.
  Gray rectangles in Figures (1)-(3) represent consistent signature sequences for occurrences of $P$. 
  (1) Each $\mathit{XShrink}_{t}^{P}$ and $\mathit{XPow}_{t}^{P}$ occur on substring $P$ 
  in $\mathit{shrink}_{t}^{T}$ and $\mathit{Pow}_{t}^{T}$, respectively, where $T = LPR$.
  (2) The substring $P$ can be represented by 
  $\hat{L}_{0}^{P}L_{0}^{P}\hat{L}_{1}^{P}L_{1}^{P}\mathit{XShrink}_{2}^{P}R_{1}^{P}\hat{R}_{1}^{P}R_{0}^{P}\hat{R}_{0}^{P}$. 
  (3) There exist common signatures on every substring $P$ in the derivation tree.
  (4) The derivation tree of $\id{s_3}$ and the subtree $X$ in the proof of Lemma~\ref{lem:concatenate1}. 
  } 
  \label{fig:CommonSequence}
\end{center}
\end{figure}

\begin{lemma}[common sequences~\cite{18045}]\label{lem:common_sequence2}
  Let $\mathcal{G} = (\Sigma, \mathcal{V}, \mathcal{D}, S)$ be
  the signature encoding of text $T$ and let $P$ be any string.
  Then there exists a common sequence $v = e_1, \ldots, e_d$ of signatures
  w.r.t. $\mathcal{G}$ which satisfies the following three conditions:
(1) $\valp{v} = P$,
(2) $|\encpow{v}| = O(\log |P| \log^* M)$, and
(3) for any $e \in \mathcal{V}$ and integer $i$ such that $\val{e}[i..i+|P|-1] = P$, $v$ occurs at position $i$ in $e$. 
\end{lemma}

\begin{proof}
We consider the following short sequence $\mathit{Uniq}(P)$ of signatures 
which represents $P$ (see also Fig.~\ref{fig:CommonSequence}):
\[ 
\mathit{Uniq}(P) = \hat{L}_{0}^{P}L_{0}^{P} \cdots 
\hat{L}_{h^{P}-1}^{P}L_{h^{P}-1}^{P}\mathit{XShrink}_{h^{P}}^{P}R_{h^{P}-1}^{P}\hat{R}_{h^{P}-1}^{P} \cdots R_{0}^{P}\hat{R}_{0}^{P},
\] 
where $h^{P}$ is the minimum integer such that 
$|\encpow{\mathit{XShrink}_{h^{P}}^{P}}| \leq \Delta_{L} + \Delta_{R} + 9$. 
We show Lemma~\ref{lem:common_sequence2} using $\mathit{Uniq}(P)$, 
namely, we show $v = \mathit{Uniq}(P)$ satisfies all conditions (1)-(3).  
(1) This follows from Definition~\ref{def:xshrink} (see also Fig.~\ref{fig:CommonSequence}(2)). 
(2) $|\encpow{\mathit{XShrink}_{h^{P}}^{P}}| \leq \Delta_{L} + \Delta_{R} + 9$, 
$h^{P} \leq \log |P|$, 
$|L_{t}^{P}| = \Theta(\Delta_{L})$, 
$|R_{t}^{P}| = \Theta(\Delta_{R})$ and 
$|\encpow{\hat{L}_{t}^{P}}| = |\encpow{\hat{R}_{t}^{P}}| = 1$ for $1 \leq t \leq h^{P}$. 
Hence $|\encpow{\mathit{Uniq}(P)}| = O(\log |P| \log^* M)$. 
(3) 
For simplicity, here we only consider the case where $\val{e} = T$,
since other cases can be shown similarly.
Consider any integer $i$ with $T[i..i+|P|-1] = P$
(see also Fig.~\ref{fig:CommonSequence}(2)). 
Note that for $0 \leq t < h^{P}$,
if $\mathit{XShrink}_{t}^{P}$ occurs in $\shrink{t}{T}$, then 
$\mathit{XPow}_{t}^{P}$ always occurs in $\pow{t}{T}$,
because $\mathit{XPow}_{t}^{P}$ is determined only by $\mathit{XShrink}_{t}^{P}$. 
Similarly, for $0 < t \leq h^{P}$, if $\mathit{XPow}_{t-1}^{P}$ occurs in $\pow{t-1}{T}$, 
then $\mathit{XShrink}_{t}^{P}$ always occurs in $\shrink{t}{T}$.
Since $\mathit{XShrink}_{0}^{P}$ occurs at position $i$ in $\shrink{0}{T}$, 
$\mathit{XShrink}_{t}^{P}$ and $\mathit{XPow}_{t}^{P}$ occur in the derivation tree of $\id{T}$. 
Hence we discuss the positions of $\mathit{XShrink}_{t}^{P}$ and $\mathit{XPow}_{t}^{P}$.
Now, let $\hat{c}_{t}$ + 1 and $c_{t}$ + 1 be the beginning positions of
the corresponding occurrence of $\mathit{XShrink}_{t}^{P}$ in $\shrink{t}{T}$ and 
that of $\mathit{XPow}_{t}^{P}$ in $\pow{t}{T}$, respectively. 
Then $\shrink{t}{T}[..\hat{c}_{t}]$ consists of $\pow{t-1}{T}[..c_{t-1}]$ and $L_{t-1}^{P}$ for $0 < t \leq h^{P}$. 
Also, $\pow{t}{T}[..c_{t}]$ consists of $\shrink{t}{T}[..\hat{c}_{t}]$ and $\hat{L}_{t}^{P}$ for $0 \leq t < h^{P}$. 
This means that $\mathit{Uniq}(P)$ occurs at position $i$ in $\id{T}$. 

Therefore Lemma~\ref{lem:common_sequence2} holds.
\qed
\end{proof}
The sequence $v$ of signatures in Lemma~\ref{lem:common_sequence2} is called
a \emph{common sequence} of $P = \val{e}[i..i+k-1]$ w.r.t. $\mathcal{G}$. 
Lemma~\ref{lem:common_sequence2} implies that any substring $P$ 
of $T$ can be represented by a sequence $p$ of signatures with
$|\encpow{p}| = O(\log |P| \log^* M)$.
The common sequences are conceptually equivalent to
the \emph{cores}~\cite{maruyama13:_esp} which are defined for the
\emph{edit sensitive parsing} of a text,
a kind of locally consistent parsing of the text.


The number of ancestors of nodes corresponding to $\mathit{Uniq}(P)$ is upper bounded by the next lemma.
\begin{lemma}\label{lem:ancestors}
  Let $T$ and $P$ be strings, and let $\mathcal{T}$ be the derivation tree of the signature encoding of $T$.
  Consider an occurrence of $P$ in $T$,
  and the induced subtree $X$ of $\mathcal{T}$
  whose root is the root of $\mathcal{T}$ and whose leaves are 
  the parents of the nodes representing $\uniq{P}$.
  Then $X$ contains $O(\log^* M)$ nodes for every height and
  $O(\log |T| + \log |P| \log^* M)$ nodes in total.
\end{lemma}
\begin{proof}
  By Definition~\ref{def:xshrink}, for every height,
  $X$ contains $O(\log^* M)$ nodes that are parents of the nodes representing $\uniq{P}$.
  Lemma~\ref{lem:ancestors} holds because the number of nodes at some height is halved when $\mathit{Shrink}$ is applied.
  More precisely, considering the $x$ nodes of $X$ at some height to which $\mathit{Shrink}$ is applied,
  the number of their parents is at most $(x + 2) / 2$.
\end{proof}

The next Lemma immediately follows from Lemma~\ref{lem:ancestors}, which will be mainly used in the proof
of Lemma~\ref{lem:upperbound_signature} in Appendix and the proof of Lemma~\ref{lem:INSERT_DELETE}.
\begin{lemma}\label{lem:concatenate1}
  Let $s_1, s_2, s_3$ be any strings
  such that $s_3 = s_1s_2$, and let $\mathcal{T}$ be the derivation tree of $\id{s_3}$. 
  Consider the induced subtree $X$ of $\mathcal{T}$
  whose root is the root of $\mathcal{T}$ and whose leaves are 
  the parents of the nodes representing $\mathit{Uniq}(s_1)\mathit{Uniq}(s_2)$
  (see also Fig.~\ref{fig:CommonSequence}(4)).
  Then the size of $X$ is $O(\log |s_3| \log^* M)$.
\end{lemma}

The following lemma is about the computation of a common sequence of $P$.
\begin{lemma}\label{lem:ComputeShortCommonSequence}
Using the DAG for a signature encoding $\mathcal{G} = (\Sigma, \mathcal{V}, \mathcal{D}, S)$ of size $w$,  
given a signature $e \in \mathcal{V}$ (and its corresponding node in the DAG)
and two integers $i$ and $k$, 
we can compute $\encpow{\mathit{Uniq}(s[i..i+k-1])}$ in $O(\log |s| + \log k \log^* M)$ time, 
where $s = \val{e}$.
\end{lemma}
\begin{proof}
  Let $v$ be the common sequence of nodes
  which represents $\mathit{Uniq}(s[i..i+k-1])$
  and occurs at position $i$ in $e$.
  Starting at the given node in the DAG which corresponds to $e$,
  we compute the induced subtree which represents $\mathit{Uniq}(s[i..i+k-1])$,
  rooted at the lowest common ancestor of the nodes in $v$.
  By Lemma~\ref{lem:ancestors},
  the size of this subtree is $O(\log |s| + \log k \log^* M)$.
  We can obtain the root of this subtree in $O(\log |s|)$ time
  from the node representing $e$.
  Hence Lemma~\ref{lem:ComputeShortCommonSequence} holds.
\qed
\end{proof}

The next lemma shows that we can compute $\LCEQ$ efficiently using the signature encoding of the (dynamic) text. 
\begin{lemma}\label{lem:sub_operation_lemma}
Using the DAG for a signature encoding $\mathcal{G} = (\Sigma, \mathcal{V}, \mathcal{D}, S)$ of size $w$, 
we can support queries $\LCEQ(s_1, s_2, i, j)$ and $\LCEQ(s_1^{R}, s_2^{R}, i, j)$ 
in $O(\log |s_1| + \log |s_2| + \log\ell \log^* M)$ time 
for given two signatures $e_1, e_2 \in \mathcal{V}$ and 
two integers $1 \leq i \leq |s_1|$, $1 \leq j \leq |s_2|$, 
where $s_1 = \val{e_1}$, $s_2 = \val{e_2}$ and $\ell$ is the answer to the $\LCEQ$ query. 
\end{lemma}
\begin{proof}
We focus on $\LCEQ(s_1, s_2, i, j)$ as $\LCEQ(s_1^{R}, s_2^{R}, i, j)$ is supported similarly.

Let $P$ denote the longest common prefix of $s_1$ and $s_2$.
Our algorithm simultaneously traverses two derivation trees rooted at $e_1$ and $e_2$
and computes $P$ by matching the common signatures greedily from left to right.
Since $\uniq{P}$ occurs at position $i$ in $e_1$ and at position $j$ in $e_2$ by Lemma~\ref{lem:common_sequence2},
we can compute $P$ by at least finding the common sequence of nodes which represents $\uniq{P}$,
and hence, we only have to traverse ancestors of such nodes.
By Lemma~\ref{lem:ancestors},
the number of nodes we traverse, which dominates the time complexity, is upper bounded by
$O(\log |s_1| + \log |s_2| + \encpow{uniq{P}}) = O(\log |s_1| + \log |s_2| + \log \ell \log^* M)$.

\qed
\end{proof}

\subsubsection{Construction}
Recall that a signature encoding $\mathcal{G}$ generating a string $T$ is represented and maintained by a data structure $\mathcal{H}$.
We show how to construct an $\mathcal{H}(\log w, w)$ or $\mathcal{H}(f_{\mathcal{A}}, f'_{\mathcal{A}})$ for $\mathcal{G}$.
It can be constructed from various types of inputs, such as 
(1) a plain (uncompressed) string $T$,
(2) the LZ77 factorization of $T$, and
(3) an SLP which represents $T$,
as summarized by the following theorem.
\begin{theorem}\label{theo:HConstructuionTheorem}
\begin{enumerate}
\item Given a string $T$ of length $N$, we can construct $\mathcal{H}(\log w, w)$ for 
the signature encoding of size $w$ which represents $T$ 
in $O(N)$ time and working space, or $\mathcal{H}(f_{\mathcal{A}}, f'_{\mathcal{A}})$ 
in $O(N f_{\mathcal{A}})$ time and $O(f'_{\mathcal{A}}+ w)$ working space. 
\item Given $f_1, \ldots, f_z$ LZ77 factors without self reference of size $z$ representing $T$ of length $N$, 
we can construct $\mathcal{H}(f_{\mathcal{A}}, f'_{\mathcal{A}})$ for the signature encoding of size $w$ which represents $T$ 
in $O(z f_{\mathcal{A}} \log N \log ^* M)$ time and $O(f'_{\mathcal{A}} + w)$ working space.
\item Given an SLP $\mathcal{S} = \{X_i \rightarrow \mathit{expr}_i\}_{i = 1}^{n}$ of size $n$ representing $T$ of length $N$,
 we can construct $\mathcal{H}(\log w, w)$ for the signature encoding of size $w$ which represents $T$ 
 in $O(n \log \log n \log N \log ^* M)$ time and $O(n \log^* M + w)$ working space, or
 $\mathcal{H}(f_{\mathcal{A}}, f'_{\mathcal{A}})$ in $O(n f_{\mathcal{A}} \log N \log ^* M)$ time and $O(f'_{\mathcal{A}} + w)$ working space. 
\end{enumerate} 
\end{theorem}
\begin{proof}
See Appendix~\ref{sec:Proof_HConstructuionTheorem}.
\end{proof}
In the static case, the $M$ term of Theorem~\ref{theo:HConstructuionTheorem}
can be replaced with $N$.

\subsubsection{Update}
In Section~\ref{sec:static_section}, we describe our dynamic index using 
$\mathcal{H}$ for a signature encoding $\mathcal{G}$ generating a string $T$.
For this end, we consider the following update operations for $\mathcal{G}$ using $\mathcal{H}$.
\begin{itemize}
 \item $\mathit{INSERT}(Y, i)$: Given a string $Y$ and an integer $i$, update $\mathcal{H}$. 
 Updated $\mathcal{H}$ handles a signature encoding $\mathcal{G}$ generates $T' = T[..i-1]Y[1..|Y|]T[i..]$. 
 \item $\mathit{DELETE}(i, k)$: Given two integers $i,k$, update $\mathcal{H}$. 
 Updated $\mathcal{H}$ handles a signature encoding $\mathcal{G}$ generates $T' = T[..i-1]T[i+k..]$. 
\end{itemize}

During updates, a new assignment $e \rightarrow \mathit{xexpr}$ is appended to $\mathcal{G}$ whenever it is needed,
in this paper, where $e = \max \mathcal{V} + 1$ that has not been used as a signature.
Specifically, we assign new signature to $\mathit{xeptr}$ when $\mathit{Sig}(\mathit{xeptr})$ returns 
undefined for some form $\mathit{xeptr}$ during updates. 
Also, updates may produce a redundant signature whose parents in the DAG are all removed.
To keep $\mathcal{G}$ admissible, we remove such redundant signatures from $\mathcal{G}$ during updates.

\begin{lemma}\label{lem:INSERT_DELETE}
Using $\mathcal{H}(f_{\mathcal{A}}, f'_{\mathcal{A}})$ 
for a signature encoding 
$\mathcal{G} = (\Sigma, \mathcal{V}, \mathcal{D}, S)$ of size $w$ which generates $T$,
we can support $\mathit{INSERT}(i,Y)$ and $\mathit{DELETE}(i,k)$ in $O(f_{\mathcal{A}}(k + \log N \log^* M))$ time, 
where $|Y| = k$.
\end{lemma}
\begin{proof}
We support $\mathit{DELETE}(i,k)$ as follows:
(1) Compute a new start variable $S' = \id{T[..i-1]T[i..]}$ 
by recomputing the new signature encoding from $\mathit{Uniq}(T[..i-1])$ and $\mathit{Uniq}(T[i+k..])$.
This can be done in $O(f_{\mathcal{A}}\log N \log^* M)$ time 
by Lemmas~\ref{lem:ComputeShortCommonSequence} and \ref{lem:concatenate1}. 
(2) Remove all redundant signatures $Z$ from $\mathcal{H}(f_{\mathcal{A}}, f'_{\mathcal{A}})$.
Note that if a signature is redundant, then all the signatures along the path from $S$ to it are also redundant.
Hence, we can remove all redundant signatures efficiently by depth-first search starting from $S$,
which takes $O(f_{\mathcal{A}}|Z|)$ time, where $|Z| = O(k + \log N \log^* M)$ by Lemma~\ref{lem:concatenate1}.

Similarly, we can compute $\mathit{INSERT}$ operation in $O(f_{\mathcal{A}}(|Y| + \log N \log^* M))$ time 
by creating $S'$ using $\mathit{Uniq}(T[..i-1])$, $\mathit{Uniq}(Y)$ and $\mathit{Uniq}(T[i+k..])$. 
Note that we can naively compute $\id{s}$ for a given string $s \in \Sigma^+$ in $O(f_{\mathcal{A}} |s|)$ time. 
Therefore Lemma~\ref{lem:INSERT_DELETE} holds.
\qed
\end{proof}

\section{Dynamic Compressed Index}\label{sec:static_section}

In this section, we present our dynamic compressed index based on signature encoding.
As already mentioned in Section~\ref{sec:related_work}, 
our strategy for pattern matching is different from that of Alstrup et al.~\cite{LongAlstrup}.
It is rather similar to the one taken in the static index for SLPs of Claude and Navarro~\cite{claudear:_self_index_gramm_based_compr}.
Besides applying their idea to run-length ACFGs, we show how to speed up pattern matching by utilizing the properties of signature encodings.

The rest of this section is organized as follows:
In Section~\ref{sec:SLPIndex}, we briefly review the idea for the SLP index of Claude and Navarro~\cite{claudear:_self_index_gramm_based_compr}.
In Section~\ref{subsec:RunLen}, we extend their idea to run-length ACFGs.
In Section~\ref{subsec:StaticIndexSigDic}, we consider an index on signature encodings and
improve the running time of pattern matching by using the properties of signature encodings.
In Section~\ref{subsec:DynIndexSigEnc}, we show how to dynamize our index.

\subsection{Static Index for SLP}\label{sec:SLPIndex}
We review how the index in~\cite{claudear:_self_index_gramm_based_compr}
for SLP $\mathcal{S}$ generating a string $T$ computes $\Occ(P,T)$ for a given string $P$.
The key observation is that, any occurrence of $P$ in $T$
can be uniquely associated with the lowest node that covers the occurrence of $P$ in the derivation tree.
As the derivation tree is binary, if $|P| > 1$, then the node is labeled with some variable $X \in \mathcal{V}$ such that
$P_1$ is a suffix of $X.{\rm left}$ and $P_2$ is a prefix of $X.{\rm right}$, where $P = P_1P_2$ with $1 \leq |P_1| < |P|$.
Here we call the pair $(X, |X.{\rm left}| - |P_1| + 1)$ a \emph{primary occurrence} of $P$.
Then, we can compute $\Occ(P, T)$ by first computing such a primary occurrence and
enumerating the occurrences of $X$ in the derivation tree.

Formally, we define the \emph{primary occurrences} of $P$ as follows.
\begin{definition}[The set of primary occurrences of $P$]\label{def:primary}
For a string $P$ with $|P| > 1$ and an integer $1 \leq j < |P|$, 
we define $\mathit{pOcc}_{\mathcal{S}}(P, j)$ and $\mathit{pOcc}_{\mathcal{S}}(P)$ as follows:
\begin{eqnarray*}
 \mathit{pOcc}_{\mathcal{S}}(P, j) & = & \{ (X, |X.{\rm left}| - j + 1) \mid X \in \mathcal{V},\\
                                   && \text{$P[..j]$ is a suffix of $X.{\rm left}$, $P[j+1..]$ is a prefix of $X.{\rm right}$}\},\\
 \mathit{pOcc}_{\mathcal{S}}(P) & = & \bigcup_{1 \leq j < |P|} \mathit{pOcc}_{\mathcal{S}}(P, j),\\
\end{eqnarray*}
We call each element of $\mathit{pOcc}_{\mathcal{S}}(P)$ a primary occurrence of $P$.
\end{definition}
The set $\mathit{Occ}(P,T)$ of occurrences of $P$ in $T$ is represented by $\mathit{pOcc}_{\mathcal{S}}(P)$ as follows.
\begin{observation}\label{ob:OccurrenceSLP}
For any string $P$,
\begin{eqnarray*}
 \mathit{Occ}(P,T) &=&
  \begin{cases}
  \{ j+k-1 \mid (X, j) \in \mathit{pOcc}_{\mathcal{S}}(P), k \in \mathit{vOcc}(X, S)\} & \mbox{ if } |P| > 1,\\
  \mathit{vOcc}(X ,S) ((X \rightarrow P) \in \mathcal{D}) & \mbox{ if } |P| = 1.\\
  \end{cases}\\
\end{eqnarray*}
\end{observation}

By Observation~\ref{ob:OccurrenceSLP},
the task is to compute $\mathit{pOcc}_{\mathcal{S}}(P)$ and $\mathit{vOcc}(X,S)$ efficiently.
Note that $\mathit{vOcc}(X,S)$ can be computed in $O(|\mathit{vOcc}(X,S)| h)$ time
by traversing the DAG in a reversed direction (i.e., using $\mathit{parents}(X)$ function recursively) from $X$ to the source,
where $h$ is the height of the derivation tree of $S$.
Hence, in what follows, we focus on how to compute $\mathit{pOcc}_{\mathcal{S}}(P)$ for a string $P$ with $|P| > 1$.
In order to compute $\mathit{pOcc}_{\mathcal{S}}(P,j)$,
we use a data structure to solve the following problem:
\begin{problem}[Two-Dimensional Orthogonal Range Reporting Problem]\label{problem:2D}
Let $\mathcal{X}$ and $\mathcal{Y}$ denote subsets of two ordered sets,
and let $\mathcal{R} \subseteq \mathcal{X} \times \mathcal{Y}$ be a set of points on the two-dimensional plane,
where $|\mathcal{X}|, |\mathcal{Y}| \in O(|\mathcal{R}|)$.
A data structure for this problem supports a query $\mathit{report}_{\mathcal{R}}(x_1,x_2,y_1,y_2)$;
given a rectangle $(x_1,x_2,y_1,y_2)$ with $x_1, x_2 \in \mathcal{X}$ and $y_1, y_2 \in \mathcal{Y}$, 
returns $\{ (x,y) \in \mathcal{R} \mid x_1 \leq x \leq x_2, y_1 \leq y \leq y_2 \}$.
\end{problem}

Data structures for Problem~\ref{problem:2D} are widely studied in computational geometry.
There is even a dynamic variant, which we finally use for our dynamic index in Section~\ref{subsec:DynIndexSigEnc}.
Until then, we just use any static data structure that occupies $O(|\mathcal{R}|)$ space and
supports queries in $O(\hat{q}_{|\mathcal{R}|} + q_{|\mathcal{R}|} \mathit{qocc})$ time with $\hat{q}_{|\mathcal{R}|} = O(\log |\mathcal{R}|)$,
where $\mathit{qocc}$ is the number of points to report.

Now, given an SLP $\mathcal{S}$, we consider a two-dimensional plane defined by 
$\mathcal{X} = \{ X.{\rm left}^{R} \mid X \in \mathcal{V} \}$ and $\mathcal{Y}  = \{ X.{\rm right} \mid X \in \mathcal{V} \}$,
where elements in $\mathcal{X}$ and $\mathcal{Y}$ are sorted by lexicographic order.
Then consider a set of points $\mathcal{R} = \{ (X.{\rm left}^{R}, X.{\rm right}) \mid X \in \mathcal{V} \}$.
For a string $P$ and an integer $1 \leq j < |P|$,
let $y_1^{(P,j)}$ (resp. $y_2^{(P,j)}$) denote the lexicographically smallest (resp. largest) element 
in $\mathcal{Y}$ that has $P[j+1..]$ as a prefix.
If there is no such element, it just returns NIL and we can immediately know that $\mathit{pOcc}_{\mathcal{S}}(P,j) = \emptyset$.
We also define $x_1^{(P,j)}$ and $x_2^{(P,j)}$ in a similar way over $\mathcal{X}$, i.e.,
$x_1^{(P,j)}$ (resp. $x_2^{(P,j)}$) is the lexicographically smallest (resp. largest) element in $\mathcal{X}$ that has $P[..j]^{R}$ as a prefix.
Then, $\mathit{pOcc}_{\mathcal{S}}(P, j)$ can be computed by a query $\mathit{report}_{\mathcal{R}}(x_1^{(P,j)}, x_2^{(P,j)}, y_1^{(P,j)}, y_2^{(P,j)})$.
See also Example~\ref{ex:SLPGrid}.
\begin{example}[SLP]\label{ex:SLPGrid}
Let $\mathcal{S}$ be the SLP of Example~\ref{ex:SLP}. 
Then, 
\begin{eqnarray*}
\mathcal{X} &=& \{ x_1, x_4, x_2, x_8, x_5, x_9, x_6, x_{10}, x_{11}, x_3, x_7 \}, \\
\mathcal{Y} &=& \{ y_1, y_8, y_2, y_7, y_9, y_3, y_4, y_{5}, y_{6}, y_{10}, y_{11} \}, \\
\end{eqnarray*}
where 
$x_i = \val{X_i}^{R}$, $y_i = \val{X_i}$ for any $X_i \in \mathcal{V}$.
Given a pattern $P = BCAB$, then 
$\mathit{pOcc}_{\mathcal{S}}(P,1) = \{ (X_{6},3), (X_{11}, 10 )\}$, 
$\mathit{pOcc}_{\mathcal{S}}(P,2) = \{ (X_{9},1) \}$, 
$\mathit{pOcc}_{\mathcal{S}}(P,3) = \phi$, 
$\mathit{vOcc}_{\mathcal{S}}(X_{6},S) = \{1,11\}$, 
$\mathit{vOcc}_{\mathcal{S}}(X_{9},S) = \{7\}$ and 
$\mathit{vOcc}_{\mathcal{S}}(X_{11},S) = \{1\}$. 
Hence $\mathit{pOcc}_{\mathcal{S}}(P) = \{ (X_{6},3), (X_{11}, 10), (X_{9},1)\}$ and 
$\mathit{Occ}(P,T) = \{3,7,10,13\}$. 
See also Fig.~\ref{fig:grid}. 
\end{example}
\begin{figure}[ht]
\begin{center}
  \includegraphics[scale=0.48]{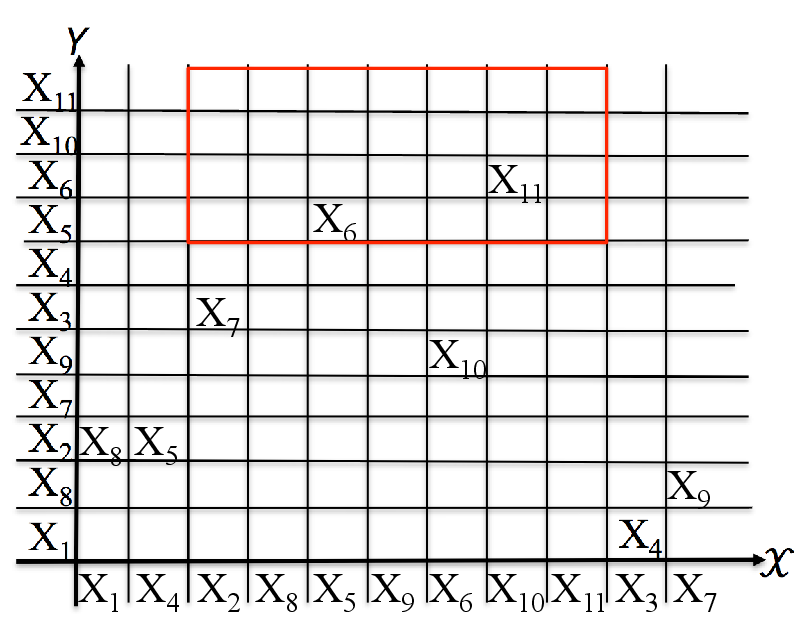}
  \caption{
  The grid represents the relation between $\mathcal{X}$, $\mathcal{Y}$ and $\mathcal{R}$ of Example~\ref{ex:SLPGrid}. 
  The red rectangle is a query rectangle $(x^{(P,1)}_1,x^{(P,1)}_2,y^{(P,1)}_1,y^{(P,1)}_2)$, where 
  $x^{(P,1)}_1 = x_2$, $x^{(P,1)}_2 = x_{11}$, $y^{(P,1)}_1 = y_5$ and $y^{(P,1)}_2 = y_{11}$. 
  Therefore, $\mathit{report}_{\mathcal{R}}(x^{(P,1)}_1,x^{(P,1)}_2,y^{(P,1)}_1,y^{(P,1)}_2) = \{ X_{6}, X_{11} \}$.
  }
  \label{fig:grid}
\end{center}
\end{figure}

We can get the following result:
\begin{lemma}\label{lem:LinearSpaceIndexOfSLP}
For an SLP $\mathcal{S}$ of size $n$,
there exists a data structure of size $O(n)$ that computes, given a string $P$,
$\mathit{pOcc}_{\mathcal{S}}(P)$ in $O(|P| (h + |P|) \log n + q_{n} |\mathit{pOcc}_{\mathcal{S}}(P)|)$ time.
\end{lemma}
\begin{proof}
For every $1 \leq j < |P|$, we compute $\mathit{pOcc}_{\mathcal{S}}(P, j)$ by 
$\mathit{report}_{\mathcal{R}}(x_1^{(P,j)}, x_2^{(P,j)}, y_1^{(P,j)}, y_2^{(P,j)})$.
We can compute $y_1^{(P,j)}$ and $y_2^{(P,j)}$ in $O((h + |P|) \log n)$ time by binary search on $\mathcal{Y}$,
where each comparison takes $O(h + |P|)$ time for expanding the first $O(|P|)$ characters of variables subjected to comparison.
In a similar way, $x_1^{(P,j)}$ and $x_2^{(P,j)}$ can be computed in $O((h + |P|) \log n)$ time.
Thus, the total time complexity is 
$O(|P| ((h + |P|) \log n + \hat{q}_{n}) + q_{n} |\mathit{pOcc}_{\mathcal{S}}(P)|) = O(|P| (h + |P|) \log n + q_{n} |\mathit{pOcc}_{\mathcal{S}}(P)|)$.
\qed
\end{proof}

 
\subsection{Static Index for Run-length ACFG}\label{subsec:RunLen}
In this subsection, we extend the idea for the SLP index described in Section~\ref{sec:SLPIndex} to run-length ACFGs.
Consider occurrences of string $P$ with $|P| > 1$ in run-length ACFG $\mathcal{G} = (\Sigma, \mathcal{V}, \mathcal{D}, S)$ generating string $T$.
The difference from SLPs is that we have to deal with occurrences of $P$
that are covered by a node labeled with $e \rightarrow \hat{e}^{d}$ but not covered by any single child of the node in the derivation tree.
In such a case, there must exist $P = P_1P_2$ with $1 \leq |P_1| < |P|$ such that
$P_1$ is a suffix of $e.{\rm left} = \valp{\hat{e}}$ and $P_2$ is a prefix of $e.{\rm right} = \valp{\hat{e}^{d-1}}$.
Let $j = |\val{\hat{e}}| - |P_1| + 1$ be a position in $\valp{\hat{e}^{d}}$ where $P$ occurs,
then $P$ also occurs at $j + c |\val{\hat{e}}|$ in $\valp{\hat{e}^{d}}$ for every positive integer
$c$ with $j + c |\val{\hat{e}}| + |P| - 1 \leq |\valp{\hat{e}^{d}}|$.
Remarking that we apply Definition~\ref{def:primary} of primary occurrences to run-length ACFGs as they are,
we formalize our observation to compute $\mathit{Occ}(P,T)$ as follows:
\begin{observation}
For any string $P$ with $|P| > 1$,
$\mathit{Occ}(P,T) = \{ j+k+c-1 \mid (e, i) \in \mathit{pOcc}_{\mathcal{G}}(P), 
c \in \{ 0 \} \cup \mathit{Run}_{\mathcal{G}}(e,j,|P|), k \in \mathit{vOcc}(e,S)\}$, where 
\begin{eqnarray*}
  \mathit{Run}_{\mathcal{G}}(e, j, |P|) & = & 
   \begin{cases}
    \{ c |\val{\hat{e}}| \mid 1 \leq c, j + c |\val{\hat{e}}| + |P| - 1 \leq |\valp{\hat{e}^{d}}| \} & \mbox{if } e \rightarrow \hat{e}^d, \\
	\emptyset & \mbox{otherwise.} \\
   \end{cases}
\end{eqnarray*}
\end{observation}
By the above observation, we can get the same result for a run-length ACFG as for an SLP in Lemma~\ref{lem:LinearSpaceIndexOfSLP}.

\subsection{Static Index for Signature Encoding}\label{subsec:StaticIndexSigDic}
We can apply the result of Section~\ref{subsec:RunLen} to signature encodings 
because signature encodings are run-length ACFGs, i.e.,
we can compute $\mathit{Occ}(P,T)$ by querying
$\mathit{report}_{\mathcal{R}}(x_1^{(P,j)}, x_2^{(P,j)},y_1^{(P,j)}, y_2^{(P,j)})$
for ``every'' $1 \leq j < |P|$.
However, the properties of signature encodings allow us to speed up pattern matching
as summarized in the following two ideas:
(1) We can efficiently compute $x_1^{(P,j)}, x_2^{(P,j)}, y_1^{(P,j)}$ and $y_2^{(P,j)}$
using LCE queries in compressed space (Lemma~\ref{lem:ComputePatternRange}).
(2) We can reduce the number of $\mathit{report}_{\mathcal{R}}$ queries from 
$O(|P|)$ to $O(\log |P| \log^* M)$ by using the property of the common sequence of $P$ (Lemma~\ref{lem:pattern_occurrence_lemma1}).

\begin{lemma}\label{lem:ComputePatternRange}
Assume that we have the DAG for a signature encoding $\mathcal{G} = (\Sigma, \mathcal{V}, \mathcal{D}, S)$ 
of size $w$ and $\mathcal{X}$ and $\mathcal{Y}$ of $\mathcal{G}$.
Given a signature $\id{P} \in \mathcal{V}$ for a string $P$ and an integer $j$, 
we can compute $x_1^{(P,j)}, x_2^{(P,j)}, y_1^{(P,j)}$ and $y_2^{(P,j)}$ in $O(\log w (\log N + \log |P| \log^* M))$ time.
\end{lemma}
\begin{proof}
By Lemma~\ref{lem:sub_operation_lemma} and Fact~\ref{obs:RandomAccess},
we can compute $x_1^{(P,j)}$ and $x_2^{(P,j)}$ 
on $\mathcal{X}$ by binary search in $O(\log w (\log N + \log |P| \log^* M))$ time.
Similarly, we can compute $y_1^{(P,j)}$ and $y_2^{(P,j)}$ in the same time. 
\qed
\end{proof}

\begin{lemma}\label{lem:pattern_occurrence_lemma1}
  Let $P$ be a string with $|P| > 1$.
  If $|\pow{0}{P}| = 1$, then $\mathit{pOcc}_{\mathcal{G}}(P) = \mathit{pOcc}_{\mathcal{G}}(P, 1)$.
  If $|\pow{0}{P}| > 1$, then $\mathit{pOcc}_{\mathcal{G}}(P) = \bigcup_{j \in \mathcal{P}} \mathit{pOcc}_{\mathcal{G}}(P, j)$, where
  $u$ is the common sequence of $P$ and $\mathcal{P} = \{ |\valp{u[1..i]}| \mid 1 \leq i < |u|, u[i] \neq u[i+1] \}$.
\end{lemma}
\begin{proof}
If $|\pow{0}{P}| = 1$, then $P = a^{|P|}$ for some character $a \in \Sigma$.
In this case, $P$ must be contained in a node labeled with a signature 
$e \rightarrow \hat{e}^{d}$ such that $\hat{e} \rightarrow a$ and $d \geq |P|$.
Hence, all primary occurrences of $P$ can be found by $\mathit{pOcc}_{\mathcal{G}}(P, 1)$.

If $|\pow{0}{P}| > 1$, we consider the common sequence $u$ of $P$.
Recall that $u$ occurs at position $j$ in $e$ for any $(e,j) \in \mathit{pOcc}(P)$ by Lemma~\ref{lem:common_sequence2}. 
Hence at least $\mathit{pOcc}_{\mathcal{G}}(P) = \bigcup_{i \in \mathcal{P'}} \mathit{pOcc}_{\mathcal{G}}(P,i)$ holds, 
where $\mathcal{P'} = \{|\valp{u[1]}|, \ldots, |\valp{u[..|u|-1]}|\}$.
Moreover, we show that $\mathit{pOcc}_{\mathcal{G}}(P,i) = \emptyset$ for any $i \in \mathcal{P'}$ with $u[i] = u[i+1]$.
Note that $u[i]$ and $u[i+1]$ are encoded into the same signature in the derivation tree of $e$, and
that the parent of two nodes corresponding to $u[i]$ and $u[i+1]$ has a signature $e'$ in the form $e' \rightarrow u[i]^{d}$.
Now assume for the sake of contradiction that $e = e'$.
By the definition of the primary occurrences, $i = 1$ must hold, and hence, $\shrink{0}{P}[1] = u[1] \in \Sigma$.
This means that $P = u[1]^{|P|}$, which contradicts $|\pow{0}{P}| > 1$.
Therefore the statement holds.
\qed
\end{proof}


Using Lemmas~\ref{lem:basic_operation_time},~\ref{lem:ComputeShortCommonSequence},~\ref{lem:ComputePatternRange} and~\ref{lem:pattern_occurrence_lemma1},
we get the following lemma:
\begin{lemma}\label{lem:static_index_lemma}
For a signature encoding $\mathcal{G}$, represented by $\mathcal{H}(f_{\mathcal{A}}, f'_{\mathcal{A}})$,
of size $w$ which generates a text $T$ of length $N$,
there exists a data structure of size $O(w)$ that computes, given a string $P$,
$\mathit{pOcc}_{\mathcal{G}}(P)$ in $O(|P| f_{\mathcal{A}} + \log w \log |P| \log^* M (\log N + \log |P| \log^* M) + q_{w} |\mathit{pOcc}_{\mathcal{S}}(P)|)$ time.
\end{lemma}
\begin{proof}
We focus on the case $|\pow{0}{P}| > 1$ as the other case is easier to be solved.
We first compute the common sequence of $P$ in $O(|P|f_{\mathcal{A}})$ time.
Taking $\mathcal{P}$ in Lemma~\ref{lem:pattern_occurrence_lemma1},
we recall that $|\mathcal{P}| = O(\log |P| \log^* M)$ by Lemma~\ref{lem:common_sequence2}.
Then, in light of Lemma~\ref{lem:pattern_occurrence_lemma1}, $\mathit{pOcc}_{\mathcal{G}}(P)$ can be obtained 
by $|\mathcal{P}| = O(\log |P| \log^* M)$ range reporting queries.
For each query, we spend $O(\log w (\log N + \log |P| \log^* M))$ time to compute $x_1^{(P,j)}, x_2^{(P,j)}, y_1^{(P,j)}$ and $y_2^{(P,j)}$ by Lemma~\ref{lem:ComputePatternRange}.
Hence, the total time complexity is $O(|P| f_{\mathcal{A}} + \log |P| \log^* M ( \log w (\log N + \log |P| \log^* M) + \hat{q}_{w}) + q_{w} |\mathit{pOcc}_{\mathcal{S}}(P)|)
= O(|P| f_{\mathcal{A}} + \log w \log |P| \log^* M (\log N + \log |P| \log^* M) + q_{w} |\mathit{pOcc}_{\mathcal{S}}(P)|)$.
\qed
\end{proof}

\subsection{Dynamic Index for Signature Encoding}\label{subsec:DynIndexSigEnc}
In order to dynamize our static index in the previous subsection,
we consider a data structure for ``dynamic'' two-dimensional orthogonal range reporting
that can support the following update operations:
\begin{itemize}
\item $\mathit{insert}_{\mathcal{R}}(p, x_{\mathit{pred}}, y_{\mathit{pred}})$:
given a point $p = (x, y)$,
$x_{\mathit{pred}} = \max \{ x' \in \mathcal{X} \mid x' \leq x \}$ and
$y_{\mathit{pred}} = \max \{ y' \in \mathcal{Y} \mid y' \leq y \}$,
insert $p$ to $\mathcal{R}$ and update $\mathcal{X}$ and $\mathcal{Y}$ accordingly.
\item $\mathit{delete}_{\mathcal{R}}(p)$: given a point 
$p = (x,y) \in \mathcal{R}$, delete $p$ from $\mathcal{R}$ and update $\mathcal{X}$ and $\mathcal{Y}$ accordingly.
\end{itemize}
We use the following data structure:
\begin{lemma}[\cite{DBLP:conf/soda/Blelloch08}]\label{lem:RangeQuery}
There exists a data structure
that supports $\mathit{report}_{\mathcal{R}}(x_1,x_2,y_1,y_2)$ 
in $O(\log |\mathcal{R}| + \mathit{occ}(\log |\mathcal{R}| / \log \log |\mathcal{R}|))$ time, 
and $\mathit{insert}_{\mathcal{R}}(p,i,j)$, 
$\mathit{delete}_{\mathcal{R}}(p)$ in amortized $O(\log |\mathcal{R}|)$ time, 
where $\mathit{occ}$ is the number of the elements to output. 
This structure uses $O(|\mathcal{R}|)$ space.~\footnote{
	The original problem considers a real plane in the paper~\cite{DBLP:conf/soda/Blelloch08}, however, 
	his solution only need to compare any two elements in $\mathcal{R}$ in constant time. 
	Hence his solution can apply to our range reporting problem by maintains $\mathcal{X}$ and $\mathcal{Y}$ 
	using the data structure of order maintenance problem proposed 
	by Dietz and Sleator~\cite{DBLP:conf/stoc/DietzS87}, which enables us to 
	compare any two elements in a list $L$ and insert/delete an element to/from $L$ in constant time.
}
\end{lemma}

Now we are ready to prove Theorem~\ref{theo:dynamic_index}.
\begin{proof}[Proof of Theorem~\ref{theo:dynamic_index}]
Our index consists of $\mathcal{H}(f_{\mathcal{A}}, f'_{\mathcal{A}})$ and a dynamic range reporting data structure $\Lambda$ of Lemma~\ref{lem:RangeQuery}
whose $\mathcal{R}$ is maintained as they are defined in the static version.
We maintain $\mathcal{X}$ and $\mathcal{Y}$ in two ways;
self-balancing binary search trees for binary search,
and Dietz and Sleator's data structures for order maintenance.
Then, primary occurrences of $P$ can be computed as described in Lemma~\ref{lem:static_index_lemma}.
Adding the $O(\mathit{occ} \log N)$ term for computing all pattern occurrences from primary occurrences,
we get the time complexity for pattern matching in the statement.
Concerning the update of our index, 
we described how to update $\mathcal{H}$ after $\mathit{INSERT}$ and $\mathit{DELETE}$ in Lemma~\ref{lem:INSERT_DELETE}.
What remains is to show how to update $\Lambda$ when a signature is inserted into or deleted from $\mathcal{V}$.
When a signature $e$ is deleted from $\mathcal{V}$, 
we first locate $e.{\rm left}^R$ on $\mathcal{X}$ and $e.{\rm right}$ on $\mathcal{Y}$,
and then execute $\mathit{delete}_{\mathcal{R}}(e.{\rm left}^R, e.{\rm right})$.
When a signature $e$ is inserted into $\mathcal{V}$, 
we first locate
$x_{\mathit{pred}} = \max \{ x' \in \mathcal{X} \mid x' \leq e.{\rm left}^R \}$ on $\mathcal{X}$ and
$y_{\mathit{pred}} = \max \{ y' \in \mathcal{Y} \mid y' \leq e.{\rm right} \}$ on $\mathcal{Y}$,
and then execute $\mathit{insert}_{\mathcal{R}}((e.{\rm left}^R, e.{\rm right}), x_{\mathit{pred}}, y_{\mathit{pred}})$.
The locating can be done by binary search on $\mathcal{X}$ and $\mathcal{Y}$
in $O(\log w \log N \log^* M)$ time as Lemma~\ref{lem:ComputePatternRange}.
In a single $\mathit{INSERT}(i,Y)$ or $\mathit{DELETE}(i,y)$ operation, 
$O(y + \log N \log^* M)$ signatures are inserted into or deleted from $\mathcal{V}$, where $|Y| = y$.
Hence we get Theorem~\ref{theo:dynamic_index}.
\qed
\end{proof}

\section{Applications} \label{sec:applications}

In this section, we present a number of applications of 
the data structures of Sections~\ref{sec:Framework} and~\ref{sec:static_section}.
Theorems~\ref{theo:lzfac} and~\ref{theo:changedSLP} are applications to 
text compression.
\begin{theorem}\label{theo:lzfac} 
    Given a string $T$ of length $N$, we can compute
    the LZ77 Factorization of $T$ in $O(N f_{\mathcal{A}} + z \log w \log ^3 N
    (\log^* N)^2)$ time and $O(w + f'_{\mathcal{A}})$ working space 
    using $\mathcal{H}(f_{\mathcal{A}}, f'_{\mathcal{A}})$ for a signature encoding of size $w$ which generates $T$, 
    where $z$ is the
    size of the LZ77 factorization of $T$ and $w = O(z \log N \log^* N)$.
\end{theorem}

\begin{theorem}\label{theo:changedSLP}
(1) Given an $\mathcal{H}(f_{\mathcal{A}}, f'_{\mathcal{A}})$ for 
  a signature encoding $\mathcal{G} = (\Sigma, \mathcal{V}, \mathcal{D}, S)$ of size $w$ which generates $T$, 
  we can compute an SLP $\mathcal{S}$ of size $O(w \log |T|)$ generating $T$ in $O(w \log |T|)$ time.
(2) Let us conduct a single $\mathit{INSERT}$ or $\mathit{DELETE}$
    operation on the string $T$
    generated by the SLP of (1).
    Let $y$ be the length
    of the substring to be inserted or deleted,
    and let $T'$ be the resulting string.
    During the above operation on the string,
    we can update, in $O((y + \log |T'| \log^* M)(f_{\mathcal{A}} + \log |T'|) )$ time,
    the SLP of (1) to an SLP $\mathcal{S}'$ of size $O(w' \log |T'|)$
    which generates $T'$, 
    where $M$ is the maximum length of the dynamic text, 
    $w'$ is the size of updated $\mathcal{G}$ which generates $T'$.
\end{theorem}

Theorems~\ref{theo:faster_LCP}-\ref{theo:ImproveDictionaryMatching}
are applications to compressed string processing (CSP),
where the task is to process a given compressed representation of string(s)
without explicit decompression.
\begin{theorem} \label{theo:faster_LCP}
      Given an SLP $\mathcal{S}$ of size $n$ generating a string of length $N$,
      we can construct, in $O(n \log \log n \log N \log^* N)$ time, 
      a data structure 
      which occupies $O(n \log N \log^* N)$ space 
      and supports $\LCPQ(X_i, X_j)$ and $\LCSQ(X_i, X_j)$ queries 
      for variables $X_i, X_j$ in $O(\log N)$ time.
The $\LCPQ(X_i, X_j)$ and $\LCSQ(X_i, X_j)$ query times can be improved 
to $O(1)$ using $O(n \log n \log N \log^* N)$ preprocessing time. 
\end{theorem}

\begin{theorem}\label{theo:smaller_LCE}
  Given an SLP $\mathcal{S}$ of size $n$ generating a string $T$ of length $N$, 
  there is a data structure 
  which occupies $O(w + n)$ space and 
  supports queries $\LCEQ(X_i,X_j,a,b)$ for
  variables $X_i,X_j$, $1 \leq a \leq |X_i|$ and $1 \leq b
  \leq |X_j|$ in $O(\log N + \log \ell \log^* N)$ time, where $w = O(z \log N \log^* N)$.
  The data structure can be constructed in
  $O(n \log\log n \log N \log^* N)$ preprocessing time and 
  $O(n \log^* N + w)$
  working space, where $z \leq n$ is the size of the LZ77 factorization of $T$ and $\ell$ is the answer of LCE query.
\end{theorem}

Let $h$ be the height of the derivation tree of a given SLP $\mathcal{S}$.
Note that $h \geq \log N$.
Matsubara et al.~\cite{matsubara_tcs2009} showed 
an $O(nh(n + h \log N))$-time $O(n(n + \log N))$-space
algorithm to compute an $O(n \log N)$-size representation of 
all palindromes in the string.
Their algorithm uses
a data structure which supports in $O(h^2)$ time,
$\LCEQ$ queries of a special form $\LCEQ(X_i, X_j, 1, p_j)$~\cite{MasamichiCPM97}.
This data structure takes $O(n^2)$ space and can be constructed in 
$O(n^2 h)$ time~\cite{lifshits07:_proces_compr_texts}. 
Using Theorem~\ref{theo:smaller_LCE}, we obtain a faster algorithm,
as follows:
\begin{theorem}\label{theo:palindrome}
Given an SLP of size $n$ generating a string of length $N$,
we can compute an $O(n \log N)$-size representation
of all palindromes in the string 
in $O(n \log^2 N \log^* N)$ time and $O(n \log^* N + w)$ space.
\end{theorem}

A non-empty string $s$ is called a Lyndon word
if $s$ is the lexicographically smallest suffix of $s$.
The Lyndon factorization of a non-empty string $w$
is a sequence of pairs $(|f_i|, p_i)$
where each $f_i$ is a Lyndon word and $p_i$ is a positive integer
such that $w = f_1^{p_1} \cdots f_m^{p_m}$
and $f_{i-1}$ is lexicographically smaller than $f_i$
for all $1 \leq i < m$.
I et al.~\cite{INIBT13} proposed a Lyndon factorization algorithm
running in $O(nh(n + \log n \log N))$ time and $O(n^2)$ space.
Their algorithm use the LCE data structure on SLPs~\cite{DBLP:journals/iandc/IMSIBTNS15} 
which requires $O(n^2 h)$ preprocessing time, $O(n^2)$ working space,
and $O(h \log N)$ time for LCE queries.
We can obtain a faster algorithm using Theorem~\ref{theo:smaller_LCE}.
\begin{theorem}\label{theo:Lyndon}
  Given an SLP of size $n$ generating a string of length $N$,
  we can compute the Lyndon factorization of the string in 
  $O(n (n + \log n \log N \log^* N))$ time and $O(n^2 + z\log N \log^* N))$ space.
\end{theorem}

We can also solve \emph{the grammar compressed dictionary matching problem}~\cite{DBLP:journals/tcs/INIBT15}
with our data structures. 
We preprocess an input dictionary SLP (DSLP) $\langle \mathcal{S}, m\rangle$
with $n$ productions that represent $m$ patterns.
Given an uncompressed text $T$, the task is to output all occurrences of the patterns in $T$. 

\begin{theorem}\label{theo:ImproveDictionaryMatching}
Given a DSLP $\langle \mathcal{S}, m\rangle$ of size $n$ 
that represents a dictionary $\Pi_{\langle\mathcal{S},m \rangle}$ for $m$ patterns of total length $N$,
we can preprocess the DSLP in $O((n \log \log n + m \log m) \log N \log^* N)$ time and $O(n \log N \log^* N)$ space
so that, given any text $T$ in a streaming fashion,
we can detect all $\mathit{occ}$ occurrences of the patterns in $T$ in $O(|T|\log m \log N \log^* N + \mathit{occ})$ time.
\end{theorem}

It was shown in~\cite{DBLP:journals/tcs/INIBT15} that we can 
construct in $O(n^4\log n)$ time a data structure of size $O(n^2\log N)$ 
which finds all occurrences of the patterns in $T$ in $O(|T|(h+m))$ time,
where $h$ is the height of
the derivation tree of DSLP $\langle \mathcal{S}, m \rangle$.
Note that our data structure of Theorem~\ref{theo:ImproveDictionaryMatching}
is always smaller, and runs faster when $h = \omega(\log m \log N \log^* N)$.


\vspace*{1pc}

\noindent \textbf{Acknowledgments.} We would like to thank Pawe{\l} Gawrychowski for drawing our attention to
the work by Alstrup et al.~\cite{LongAlstrup,DBLP:conf/soda/AlstrupBR00}
and for fruitful discussions.

\bibliographystyle{splncs03}
\bibliography{ref}

\begin{thebibliography}{10}
\providecommand{\url}[1]{\texttt{#1}}
\providecommand{\urlprefix}{URL }

\bibitem{DBLP:journals/comgeo/AgarwalAG0Y13}
Agarwal, P.K., Arge, L., Govindarajan, S., Yang, J., Yi, K.: Efficient external
  memory structures for range-aggregate queries. Comput. Geom.  46(3),
  358--370 (2013), \url{http://dx.doi.org/10.1016/j.comgeo.2012.10.003}

\bibitem{LongAlstrup}
Alstrup, S., Brodal, G.S., Rauhe, T.: Dynamic pattern matching. Tech. rep.,
  Department of Computer Science, University of Copenhagen (1998)

\bibitem{DBLP:conf/soda/AlstrupBR00}
Alstrup, S., Brodal, G.S., Rauhe, T.: Pattern matching in dynamic texts. In:
  Proc. SODA 2000. pp. 819--828 (2000)

\bibitem{DBLP:journals/jcss/BeameF02}
Beame, P., Fich, F.E.: Optimal bounds for the predecessor problem and related
  problems. J. Comput. Syst. Sci.  65(1),  38--72 (2002),
  \url{http://dx.doi.org/10.1006/jcss.2002.1822}

\bibitem{bille13:_finger_compr_strin}
Bille, P., Cording, P.H., G{\o}rtz, I.L., Sach, B., Vildh{\o}j, H.W., Vind, S.:
  Fingerprints in compressed strings. In: Proc. WADS 2013. pp. 146--157 (2013)

\bibitem{BilleCCG15}
Bille, P., Christiansen, A.R., Cording, P.H., G{\o}rtz, I.L.: Finger search,
  random access, and longest common extensions in grammar-compressed strings.
  CoRR  abs/1507.02853 (2015)

\bibitem{DBLP:conf/soda/Blelloch08}
Blelloch, G.E.: Space-efficient dynamic orthogonal point location, segment
  intersection, and range reporting. In: Teng, S.H. (ed.) SODA. pp. 894--903.
  SIAM (2008)

\bibitem{claudear:_self_index_gramm_based_compr}
Claude, F., Navarro, G.: Self-indexed grammar-based compression. Fundamenta
  Informaticae  111(3),  313--337 (2011)

\bibitem{DBLP:conf/stoc/DietzS87}
Dietz, P.F., Sleator, D.D.: Two algorithms for maintaining order in a list. In:
  Aho, A.V. (ed.) Proceedings of the 19th Annual {ACM} Symposium on Theory of
  Computing, 1987, New York, New York, {USA}. pp. 365--372. {ACM} (1987),
  \url{http://doi.acm.org/10.1145/28395.28434}

\bibitem{FGGK15}
Fischer, J., Gagie, T., Gawrychowski, P., Kociumaka, T.: Approximating {LZ77}
  via small-space multiple-pattern matching. In: ESA 2015. pp. 533--544 (2015)

\bibitem{DBLP:journals/corr/abs-1107-2729}
Goto, K., Maruyama, S., Inenaga, S., Bannai, H., Sakamoto, H., Takeda, M.:
  Restructuring compressed texts without explicit decompression. CoRR
  abs/1107.2729 (2011)

\bibitem{SuperSort}
Han, Y.: Deterministic sorting in {$O (n \log \log n)$} time and linear space.
  Proc. STOC 2002 pp. 602--608 (2002)

\bibitem{DBLP:conf/dcc/HonLSSY04}
Hon, W., Lam, T.W., Sadakane, K., Sung, W., Yiu, S.: Compressed index for
  dynamic text. In: DCC 2004. pp. 102--111 (2004)

\bibitem{INIBT13}
I, T., Nakashima, Y., Inenaga, S., Bannai, H., Takeda, M.: Faster {L}yndon
  factorization algorithms for {SLP} and {LZ}78 compressed text. In: Proc.
  SPIRE. pp. 174--185 (2013)

\bibitem{IMSIBTNS15}
I, T., Matsubara, W., Shimohira, K., Inenaga, S., Bannai, H., Takeda, M.,
  Narisawa, K., Shinohara, A.: Detecting regularities on grammar-compressed
  strings. Inf. Comput.  240,  74--89 (2015)

\bibitem{DBLP:journals/iandc/IMSIBTNS15}
I, T., Matsubara, W., Shimohira, K., Inenaga, S., Bannai, H., Takeda, M.,
  Narisawa, K., Shinohara, A.: Detecting regularities on grammar-compressed
  strings. Inf. Comput.  240,  74--89 (2015),
  \url{http://dx.doi.org/10.1016/j.ic.2014.09.009}

\bibitem{DBLP:journals/tcs/INIBT15}
I, T., Nishimoto, T., Inenaga, S., Bannai, H., Takeda, M.: Compressed automata
  for dictionary matching. Theor. Comput. Sci.  578,  30--41 (2015),
  \url{http://dx.doi.org/10.1016/j.tcs.2015.01.019}

\bibitem{KiefferY00}
Kieffer, J.C., Yang, E.: Grammar-based codes: {A} new class of universal
  lossless source codes. {IEEE} Transactions on Information Theory  46(3),
  737--754 (2000)

\bibitem{lifshits07:_proces_compr_texts}
Lifshits, Y.: Processing compressed texts: A tractability border. In: Proc. CPM
  2007. LNCS, vol. 4580, pp. 228--240 (2007)

\bibitem{maruyama13:_esp}
Maruyama, S., Nakahara, M., Kishiue, N., Sakamoto, H.: Esp-index: A compressed
  index based on edit-sensitive parsing. J. Discrete Algorithms  18,  100--112
  (2013)

\bibitem{matsubara_tcs2009}
Matsubara, W., Inenaga, S., Ishino, A., Shinohara, A., Nakamura, T., Hashimoto,
  K.: Efficient algorithms to compute compressed longest common substrings and
  compressed palindromes. Theor. Comput. Sci.  410(8--10),  900--913 (2009)

\bibitem{DBLP:journals/algorithmica/MehlhornSU97}
Mehlhorn, K., Sundar, R., Uhrig, C.: Maintaining dynamic sequences under
  equality tests in polylogarithmic time. Algorithmica  17(2),  183--198 (1997)

\bibitem{MasamichiCPM97}
Miyazaki, M., Shinohara, A., Takeda, M.: An improved pattern matching algorithm
  for strings in terms of straight-line programs. In: Proc. CPM 1997. pp. 1--11
  (1997)

\bibitem{DBLP:journals/corr/MunroNV15}
Munro, J.I., Nekrich, Y., Vitter, J.S.: Dynamic data structures for document
  collections and graphs. CoRR  abs/1503.05977 (2015)

\bibitem{PolicritiP15}
Policriti, A., Prezza, N.: Fast online {Lempel-Ziv} factorization in compressed
  space. In: SPIRE 2015. pp. 13--20 (2015)

\bibitem{18045}
Sahinalp, S.C., Vishkin, U.: Data compression using locally consistent parsing.
  TechnicM report, University of Maryland Department of Computer Science
  (1995)

\bibitem{DBLP:conf/focs/SahinalpV96}
Sahinalp, S.C., Vishkin, U.: Efficient approximate and dynamic matching of
  patterns using a labeling paradigm (extended abstract). In: FOCS. pp.
  320--328. IEEE Computer Society (1996)

\bibitem{DBLP:journals/ieicet/SakamotoMKS09}
Sakamoto, H., Maruyama, S., Kida, T., Shimozono, S.: A space-saving
  approximation algorithm for grammar-based compression. IEICE Transactions
  92-D(2),  158--165 (2009)

\bibitem{DBLP:journals/jda/SalsonLLM10}
Salson, M., Lecroq, T., L{\'{e}}onard, M., Mouchard, L.: Dynamic extended
  suffix arrays. J. Discrete Algorithms  8(2),  241--257 (2010)

\bibitem{DBLP:conf/wea/TakabatakeTS14}
Takabatake, Y., Tabei, Y., Sakamoto, H.: Improved esp-index: A practical
  self-index for highly repetitive texts. In: Proc. SEA 2014. pp. 338--350
  (2014)

\bibitem{TakabatakeTS15}
Takabatake, Y., Tabei, Y., Sakamoto, H.: Online self-indexed grammar
  compression. In: SPIRE 2015. pp. 258--269 (2015)

\bibitem{LZ77}
Ziv, J., Lempel, A.: A universal algorithm for sequential data compression.
  IEEE Transactions on Information Theory  IT-23(3),  337--349 (1977)

\end{thebibliography}


\begin{thebibliography}{1}
\providecommand{\url}[1]{\texttt{#1}}
\providecommand{\urlprefix}{URL }

\bibitem{DBLP:journals/jal/BenderFPSS05}
Bender, M.A., Farach{-}Colton, M., Pemmasani, G., Skiena, S., Sumazin, P.:
  Lowest common ancestors in trees and directed acyclic graphs. J. Algorithms
  57(2),  75--94 (2005)

\bibitem{rytter03:_applic_lempel_ziv}
Rytter, W.: Application of {L}empel-{Z}iv factorization to the approximation of
  grammar-based compression. Theor. Comput. Sci.  302(1--3),  211--222 (2003)

\end{thebibliography}

\clearpage
\appendix
\section{Appendix : Theorem~\ref{theo:HConstructuionTheorem}}\label{sec:Proof_HConstructuionTheorem}
\subsection{Proof of Theorem~\ref{theo:HConstructuionTheorem}~(1)}\label{sec:Proof_HConstructuionTheorem1}
\subsubsection{$\mathcal{H}(\log w, w)$ construction in $O(N)$ time and space.}
Our algorithm computes signatures level by level, i.e.,
constructs incrementally $\shrink{0}{T}, \pow{0}{T}$, $\ldots, \shrink{h}{T}, \pow{h}{T}$.
For each level, we determine signatures by sorting signature blocks (or run-length encoded signatures) to which we give signatures.
The following two lemmas describe the procedure.

\begin{lemma} \label{lem:signature_encblock}
Given $\encblock{\pow{t-1}{T}}$ for $0 < t \leq h$, 
we can compute $\shrink{t}{T}$ in $O((b-a)+|\pow{t-1}{T}|)$ time and space,
where $b$ is the maximum integer in $\pow{t-1}{T}$ 
and $a$ is the minimum integer in $\pow{t-1}{T}$.
\end{lemma}
\begin{proof}
Since we assign signatures to signature blocks and run-length signatures in the derivation tree of $S$ in the order they appear in the signature encoding.
$\pow{t-1}{T}[i] - a$ fits in an entry of a bucket of size $b-a$ for each element of $\pow{t-1}{T}[i]$ of $\pow{t-1}{T}$.
Also, the length of each block is at most four.
Hence we can sort all the blocks of $\encblock{\pow{t-1}{T}}$ by bucket sort in $O((b-a)+|\pow{t-1}{T}|)$ time and space. 
Since $\mathit{Sig}$ is an injection
and since we process the levels in increasing order,
for any two different levels $0 \leq t' < t \leq h$,
no elements of $\shrink{t-1}{T}$ appear in $\shrink{t'-1}{T}$,
and hence no elements of $\pow{t-1}{T}$ appear in $\pow{t'-1}{T}$.
Thus, we can determine a new signature for each block in $\encblock{\pow{t-1}{T}}$,
\emph{without} searching existing signatures in the lower levels.
This completes the proof.
\end{proof}

\begin{lemma} \label{lem:signature_encpow}
Given $\encpow{\shrink{t}{T}}$, we can compute $\pow{t}{T}$ in 
$O(x + (b-a) +|\encpow{\shrink{t}{T}}|)$ time and space, 
where $x$ is the maximum length of runs in $\encpow{\shrink{t}{T}}$,
$b$ is the maximum integer in $\pow{t-1}{T}$, 
and $a$ is the minimum integer in $\pow{t-1}{T}$.
\end{lemma}
\begin{proof}
We first sort all the elements of $\encpow{\shrink{t}{T}}$ by bucket sort
in $O(b-a + |\encpow{\shrink{t}{T}}|)$ time and space, ignoring the powers of runs.
Then, for each integer $r$ appearing in $\shrink{t}{T}$,
we sort the runs of $r$'s by bucket sort with a bucket of size $x$.
This takes a total of $O(x + |\encpow{\shrink{t}{T}}|)$ time and space
for all integers appearing in $\shrink{t}{T}$.
The rest is the same as the proof of Lemma~\ref{lem:signature_encblock}.
\end{proof}

The next lemma shows how to construct $\mathcal{H}(\log w, w)$ from a sorted assignment set $\mathcal{D}$ of $\mathcal{G}$.
\begin{lemma} \label{lem:linear_construction}
Given a sorted assignment set $\mathcal{D}$ of $\mathcal{G}$, 
we can construct $\mathcal{H}(\log w, w)$ of $\mathcal{G}$ in $O(|\mathcal{V}|)$ time.
\end{lemma}
\begin{proof}
Recall that $\mathcal{H}$ consists of $\mathcal{A}$ and DAG $\mathcal{B}$.
Clearly, given a sorted assignment set $\mathcal{D}$, we can construct $\mathcal{B}$ in linear time and space.
Also, we can construct, in linear time and space, a balanced search tree for $\mathcal{A}$ from $\mathcal{D}$.
Hence Lemma~\ref{lem:linear_construction} holds.
\end{proof}

We are ready to prove the theorem.
\begin{proof}
In the derivation tree of $\id{T}$, since the number of nodes in some level is halved when going up two levels higher, every node of 
Since the size of the derivation tree of $\id{T}$ is $O(N)$,
by Lemmas~\ref{lem:CoinTossing}, \ref{lem:signature_encblock}, and~\ref{lem:signature_encpow},
we can compute $\id{T}$ and a sorted assignment set $\mathcal{D}$ of $\mathcal{G}$ in $O(N)$ time and space.
Finally, by Lemma~\ref{lem:linear_construction}, we can get $\mathcal{H}(\log w, w)$ for $\mathcal{G}$ in $O(N)$ time.
\end{proof}

\subsubsection{$\mathcal{H}(f_{\mathcal{A}}, f'_{\mathcal{A}})$ construction in $O(N f_{\mathcal{A}})$ time and $O(f'_{\mathcal{A}} + w)$ working space.}
\begin{proof}
Note that we can naively compute $\id{T}$ for a given string $T$ in 
$O(N f_{\mathcal{A}})$ time and $O(N)$ working space.
In order to reduce the working space,
we consider factorizing $T$ into blocks of size $B$ and processing them incrementally:
Starting with the empty signature encoding $\mathcal{G}$,
we can compute $\id{T}$ in $O(\frac{N}{B}f_{\mathcal{A}}(\log N \log^* M + B))$ time 
and $O(w + B + f'_{\mathcal{A}})$ working space
by using $\mathit{INSERT}(T[(i-1)B+1..iB],(i-1)B+1)$ for $i = 1, \ldots, {\frac{N}{B}}$ in increasing order.
Hence our proof is finished by choosing $B = \log N \log^* M$.
\end{proof}

\subsection{Proof of Theorem~\ref{theo:HConstructuionTheorem}~(2)}
\begin{proof}
Consider $\mathcal{H}(f_{\mathcal{A}}, f'_{\mathcal{A}})$ for an empty signature encodings $\mathcal{G}$. 
If we can compute $\mathit{INSERT}(Y,i)$ operation in $O(f_{\mathcal{A}}\log N \log^* M)$ time, 
then Theorem~\ref{theo:HConstructuionTheorem}~(2) immediately holds 
by computing $\mathit{INSERT}(f_i,|f_1 \cdots f_{i-1}|+1)$ for $1 \leq i \leq z$ incrementally. 
By the proof of Lemma~\ref{lem:INSERT_DELETE}, 
we can compute $\mathit{INSERT}(Y,i)$ for a given $\uniq{Y}$ in $O(f_{\mathcal{A}}\log N \log^* M)$ time. 
We can compute each $\uniq{f_i}$ in $O(\log N \log^* M)$ time by Lemma~\ref{lem:ComputeShortCommonSequence} 
because $f_i$ occurs previously in $T$ when $|f_i| > 1$. 
Hence we get Theorem~\ref{theo:HConstructuionTheorem}~(2).
\end{proof}
Note that we can directly show Lemma~\ref{lem:upperbound_signature} from the above proof 
because the size of $\mathcal{G}$ increases $O(\log N \log^* M)$ by Lemma~\ref{lem:concatenate1}, 
every time we do $\mathit{INSERT}(f_i,|f_1 \cdots f_{i-1}|+1)$ for $1 \leq i \leq z$.

\subsection{Proof of Theorem~\ref{theo:HConstructuionTheorem}~(3)}
\subsubsection{$\mathcal{H}(f_{\mathcal{A}},f'_{\mathcal{A}})$ construction in $O(n f_{\mathcal{A}} \log N \log^* M)$ time and $O(f'_{\mathcal{A}} + w)$ working space.}
\begin{proof}
We can construct $\mathcal{H}(f_{\mathcal{A}},f'_{\mathcal{A}})$ by $O(n)$ $\mathit{INSERT}$ operations 
as the proof of Theorem~\ref{theo:HConstructuionTheorem}~(2).
\end{proof}

\subsubsection{$\mathcal{H}(\log w, w)$ construction in $O(n \log \log n \log N \log^* M)$ time and $O(n \log^* M + w)$ working space.}\label{sec:HConstruction3-2}
In this section, we sometimes abbreviate $\val{X}$ as $X$ for $X \in \mathcal{S}$.
For example, $\shrink{t}{X}$ and $\pow{t}{X}$ represents $\shrink{t}{\val{X}}$ and $\pow{t}{\val{X}}$ respectively. 

Our algorithm computes signatures level by level, i.e.,
constructs incrementally $\shrink{0}{X_n}, \pow{0}{X_n}$, $\ldots, \shrink{h}{X_n}, \pow{h}{X_n}$.
Like the algorithm described in Section~\ref{sec:Proof_HConstructuionTheorem1}, we can create signatures 
by sorting blocks of signatures or run-length encoded signatures in the same level.
The main difference is that we now utilize the structure of the SLP, 
which allows us to do the task efficiently in $O(n \log^* M + w)$ working space.
In particular, although $|\shrink{t}{X_n}|, |\pow{t}{X_n}| = O(N)$ for $0 \leq t \leq h$, 
they can be represented in $O(n \log^* M)$ space.

In so doing, we introduce some additional notations relating to $\xshrink{t}{P}$ and $\xpow{t}{P}$ in Definition~\ref{def:xshrink}.
By Lemma~\ref{lem:common_sequence2}, for any string $P = P_1P_2$ the following equation holds:
\begin{eqnarray*}
 \xshrink{t}{P} &=& \hat{y}^{P_1}_{t} \hat{z}^{(P_1,P_2)}_t \hat{y}^{P_2}_{t} \mbox{ for } 0 < t \leq h^{P}, \\
 \xpow{t}{P} &=& y^{P_1}_{t} z^{(P_1,P_2)}_t y^{P_2}_{t} \mbox{ for } 0 \leq t < h^{P},
\end{eqnarray*}
where we define $\hat{y}^{P}_{t}$ and $y^{P}_{t}$ for a string $P$ as follows:
\begin{eqnarray*}
\hat{y}^{P}_{t} &=& 
  \begin{cases}
  \xshrink{t}{P} &\mbox{ for } 0 <  t \leq h^{P},\\
  \varepsilon &\mbox{ for } t > h^{P},\\
  \end{cases} \\
y^{P}_{t} &=& 
  \begin{cases}
  \xpow{t}{P} &\mbox{ for } 0 \leq t < h^{P},\\
  \varepsilon &\mbox{ for } t \geq h^{P}.\\
  \end{cases}
\end{eqnarray*}
For any variable $X_i \rightarrow X_{\ell} X_{r}$,
we denote $\hat{z}^{X_i}_{t} = \hat{z}^{(\val{X_{\ell}},\val{X_{r}})}_{t}$ (for $0 < t \leq h^{\val{X_i}}$)
and $z^{X_i}_{t} = z^{(\val{X_{\ell}},\val{X_{r}})}_{t}$ (for $0 \leq t < h^{\val{X_i}}$).
Note that $|z^{X_i}_{t}|, |\hat{z}^{X_i}_{t}| = O(\log^* N)$ by Lemma~\ref{lem:ancestors}.
We can use $\hat{z}_{t}^{X_1}, \ldots, \hat{z}_{t}^{X_n}$ (resp. $z_{t}^{X_1}, \ldots, z_{t}^{X_n}$) 
as a compressed representation of $\xshrink{t}{X_n}$ (resp. $\xpow{t}{X_n}$) based on the SLP:
Intuitively, $\hat{z}_{t}^{X_n}$ (resp. $z_{t}^{X_n}$) covers the middle part of $\xshrink{t}{X_n}$ (resp. $\xpow{t}{X_n}$) and
the remaining part is recovered by investigating the left/right child recursively (see also Figure.~\ref{fig:LeftRightAccessFact}).
Hence, with the DAG structure of the SLP, $\xshrink{t}{X_n}$ and $\xpow{t}{X_n}$ can 
be represented in $O(n \log^* M)$ space.
\begin{figure}[t]
\begin{center}
  \includegraphics[scale=0.75]{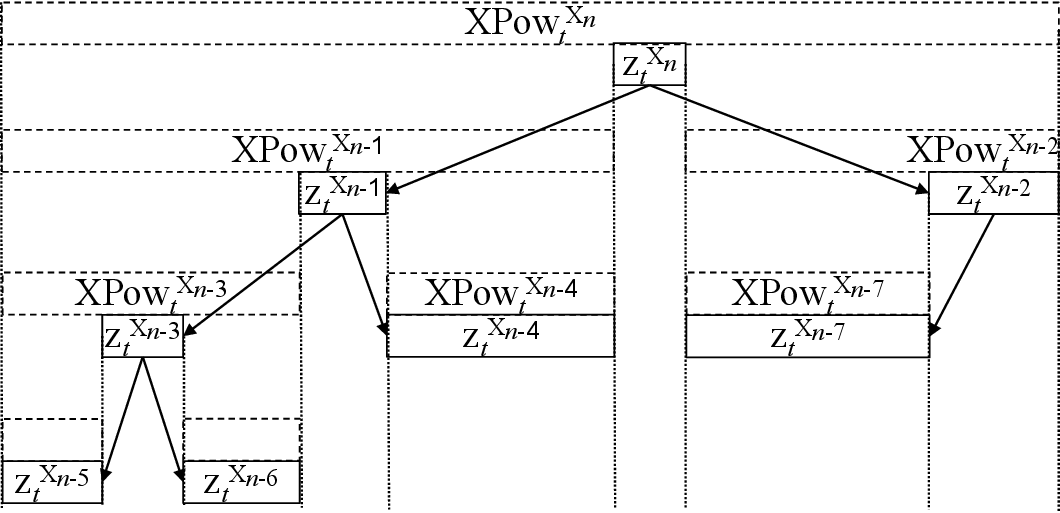}
  \caption{
  $\xpow{t}{X_n}$ can be represented by $z^{X_1}_{t}, \ldots, z^{X_n}_{t}$.  
  } 
  \label{fig:LeftRightAccessFact}
\end{center}
\end{figure}

In addition, we define $\hat{A}^{P}_{t}$, $\hat{B}^{P}_{t}$, $A^{P}_t$ and $B^{P}_t$ as follows:
For $0 < t \leq h^{P}$, $\hat{A}^{P}_t$ (resp. $\hat{B}^{P}_t$) is a prefix (resp. suffix) of $\shrink{t}{P}$ 
which consists of signatures of $A^{P}_{t-1}L^{P}_{t-1}$ (resp. $R^{P}_{t-1}B^{P}_{t-1}$); and
for $0 \leq t < h^{P}$, $A^{P}_t$ (resp. $B^{P}_t$) is a prefix (resp. suffix) of $\pow{t}{P}$ 
which consists of signatures of $\hat{A}^{P}_{t}\hat{L}^{P}_{t}$ (resp. $\hat{R}^{P}_{t}\hat{B}^{P}_{t}$).
By the definition,
$\shrink{t}{P} = \hat{A}^{P}_t\xshrink{t}{P}\hat{B}^{P}_t$ for $0 \leq t \leq h^{P}$, and
$\pow{t}{P} = A^{P}_t\xpow{t}{P}B^{P}_t$ for $0 \leq t < h^{P}$.
See Figure~\ref{fig:AXShrinkB} for the illustration.
\begin{figure}[t]
\begin{center}
  \includegraphics[scale=0.75]{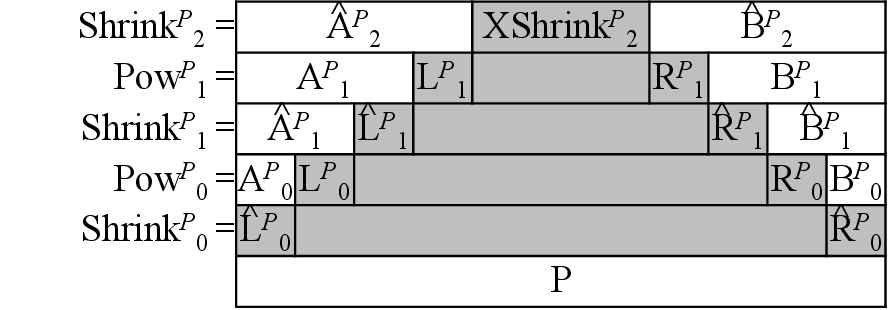}
  \caption{
  An abstract image of $\shrink{t}{P}$ and $\pow{t}{P}$ for a string $P$. 
  For $0 \leq t < h^{P}$, $A^{P}_{t}L^{P}_{t}$ (resp. $R^{P}_{t}B^{P}_{t}$) 
  is encoded into $\hat{A}^{P}_{t+1}$ (resp. $\hat{B}^{P}_{t+1}$). 
  Similarly, for $0 < t < h^{P}$, $\hat{A}^{P}_{t}\hat{L}^{P}_{t}$ (resp. $\hat{R}^{P}_{t}\hat{B}^{P}_{t}$) 
  is encoded into $A^{P}_{t}$ (resp. $B^{P}_{t}$). 
  } 
  \label{fig:AXShrinkB}
\end{center}
\end{figure}

Since $\shrink{t}{X_n} = \hat{A}_{t}^{X_n} \xshrink{t}{X_n} \hat{B}_{t}^{X_n}$ for $0 < t \leq h^{X_n}$,
we use $\hat{\Lambda}_{t} = (\hat{z}_{t}^{X_1}, \ldots, \hat{z}_{t}^{X_n}, \hat{A}^{X_n}_{t}, \hat{B}^{X_n}_{t})$ 
as a compressed representation of $\shrink{t}{X_n}$ of size $O(n \log^* M)$.
Similarly, for $0 \leq t < h^{X_n}$,
we use $\Lambda_{t} = (z_{t}^{X_1}, \ldots, z_{t}^{X_n}, A^{X_n}_{t}, B^{X_n}_{t})$ 
as a compressed representation of $\pow{t}{X_n}$ of size $O(n \log^* M)$.

Our algorithm computes incrementally $\Lambda_{0}, \hat{\Lambda}_{1}, \ldots, \hat{\Lambda}_{h^{X_n}}$.
Note that, given $\hat{\Lambda}_{h^{X_n}}$,
we can easily get $\pow{h^{X_n}}{X_n}$ of size $O(\log^* M)$ in $O(n \log^* M)$ time,
and then $\id{\val{X_n}}$ in $O(\log^* M)$ time from $\pow{h^{X_n}}{X_n}$.
Hence, in the following three lemmas, we show how to compute $\Lambda_{0}, \hat{\Lambda}_{1}, \ldots, \hat{\Lambda}_{h^{X_n}}$.

\begin{lemma}\label{lem:Lambda_0}
Given an SLP of size $n$, we can compute $\Lambda_{0}$ in $O(n \log \log n \log^* M)$ time and $O(n \log^*M)$ space.
\end{lemma}
\begin{proof}
We first compute, for all variables $X_i$,
$\encpow{\xshrink{0}{X_i}}$ if $|\encpow{\xshrink{0}{X_i}}| \leq \Delta_{L} + \Delta_{R} + 9$,
otherwise $\encpow{\hat{L}_{0}^{X_i}}$ and $\encpow{\hat{R}_{0}^{X_i}}$.
The information can be computed in $O(n \log^*M)$ time and space in a bottom-up manner, i.e., by processing variables in increasing order.
For $X_i \rightarrow X_{\ell} X_{r}$, if both $|\encpow{\xshrink{0}{X_{\ell}}}|$ and $|\encpow{\xshrink{0}{X_{r}}}|$ are no greater than $\Delta_{L} + \Delta_{R} + 9$,
we can compute $\encpow{\xshrink{0}{X_i}}$ in $O(\log^* M)$ time by naively concatenating $\encpow{\xshrink{0}{X_{\ell}}}$ and $\encpow{\xshrink{0}{X_{r}}}$.
Otherwise $|\encpow{\xshrink{0}{X_i}}| > \Delta_{L} + \Delta_{R} + 9$ must hold, and 
$\encpow{\hat{L}_{0}^{X_i}}$ and $\encpow{\hat{R}_{0}^{X_i}}$ can be computed in $O(1)$ time from the information for $X_{\ell}$ and $X_{r}$.

The run-length encoded signatures represented by $z_{0}^{X_i}$ can be obtained by using the above information for $X_{\ell}$ and $X_r$ in $O(\log^* M)$ time:
$z_{0}^{X_i}$ is created over run-length encoded signatures
$\encpow{\xshrink{0}{X_{\ell}}}$ (or $\encpow{\hat{R}_{0}^{X_{\ell}}}$) followed by $\encpow{\xshrink{0}{X_r}}$ (or $\encpow{\hat{R}_{0}^{X_r}}$).
Also, by definition $A_{0}^{X_n}$ and $B_{0}^{X_n}$ represents $\encpow{\hat{L}_{0}^{X_n}}$ and $\encpow{\hat{R}_{0}^{X_n}}$, respectively.

Hence, we can compute in $O(n \log^* M)$ time $O(n \log^*M)$ run-length encoded signatures to which we give signatures.
We determine signatures by sorting the run-length encoded signatures as Lemma~\ref{lem:signature_encpow}.
However, in contrast to Lemma~\ref{lem:signature_encpow},
we do not use bucket sort for sorting the powers of runs
because the maximum length of runs could be as large as $N$ and we cannot afford $O(N)$ space for buckets.
Instead, we use the sorting algorithm of Han~\cite{SuperSort} which sorts $x$ integers in $O(x \log\log x)$ time and $O(x)$ space.
Hence, we can compute $\Lambda_{0}$ in $O(n \log \log n \log^* M)$ time and $O(n \log^*M)$ space.
%
%
%
\end{proof}

\begin{lemma}\label{lem:Lambda_t}
Given $\hat{\Lambda}_{t}$, we can compute $\Lambda_{t}$ in $O(n \log \log n \log^*M)$ time and $O(n \log^*M)$ space.
\end{lemma}
\begin{proof}
The computation process is similar to that of Lemma~\ref{lem:Lambda_0},
except that we also use the information in $\hat{\Lambda}_{t}$.

We first compute, for all variables $X_i$,
$\encpow{\xshrink{t}{X_i}}$ if $|\encpow{\xshrink{t}{X_i}}| \leq \Delta_{L} + \Delta_{R} + 9$,
otherwise $\encpow{\hat{L}_{t}^{X_i}}$ and $\encpow{\hat{R}_{t}^{X_i}}$.
The information can be computed in $O(n \log^*M)$ time and space in a bottom-up manner, i.e., by processing variables in increasing order.
For $X_i \rightarrow X_{\ell} X_{r}$, if both $|\encpow{\xshrink{t}{X_{\ell}}}|$ and $|\encpow{\xshrink{t}{X_{r}}}|$ are no greater than $\Delta_{L} + \Delta_{R} + 9$,
we can compute $\encpow{\xshrink{0}{X_i}}$ in $O(\log^* M)$ time 
by naively concatenating $\encpow{\xshrink{t}{X_{\ell}}}$, $\encpow{\hat{z}_{t}^{X_i}}$ and $\encpow{\xshrink{t}{X_{r}}}$.
Otherwise $|\encpow{\xshrink{t}{X_i}}| > \Delta_{L} + \Delta_{R} + 9$ must hold, and 
$\encpow{\hat{L}_{0}^{X_i}}$ and $\encpow{\hat{R}_{0}^{X_i}}$ can be computed in $O(1)$ time from $\encpow{\hat{z}_{t}^{X_i}}$ and the information for $X_{\ell}$ and $X_{r}$.

The run-length encoded signatures represented by $z_{t}^{X_i}$ can be obtained in $O(\log^* M)$ time
by using $\hat{z}_{t}^{X_i}$ and the above information for $X_{\ell}$ and $X_r$:
$z_{t}^{X_i}$ is created over run-length encoded signatures that are obtained by concatenating 
$\encpow{\xshrink{0}{X_{\ell}}}$ (or $\encpow{\hat{R}_{0}^{X_{\ell}}}$), $z_{t}^{X_i}$ and $\encpow{\xshrink{0}{X_r}}$ (or $\encpow{\hat{R}_{0}^{X_r}}$).
Also, $A_{t}^{X_n}$ and $B_{t}^{X_n}$ represents $\hat{A}_{t}^{X_n} \hat{L}_{t}^{X_n}$ and $\hat{R}_{t}^{X_n} \hat{B}_{t}^{X_n}$, respectively.

Hence, we can compute in $O(n \log^* M)$ time $O(n \log^*M)$ run-length encoded signatures to which we give signatures.
We determine signatures in $O(n \log \log n \log^* M)$ time by sorting the run-length encoded signatures as Lemma~\ref{lem:Lambda_t}.
\end{proof}

\begin{lemma}\label{lem:hat_Lambda_t}
Given $\Lambda_{t}$, we can compute $\hat{\Lambda}_{t+1}$ in $O(n \log^*M)$ time and $O(n \log^*M)$ space.
\end{lemma}
\begin{proof}
In order to compute $\hat{z}_{t+1}^{X_i}$ for a variable $X_i \rightarrow X_{\ell} X_{r}$,
we need a signature sequence on which $\hat{z}_{t+1}^{X_i}$ is created,
as well as its context, i.e., $\Delta_{L}$ signatures to the left and $\Delta_{R}$ to the right.
To be precise, the needed signature sequence is $v_{t}^{X_{\ell}} z_{t}^{X_i} u_{t}^{X_{r}}$,
where $u_{t}^{X_j}$ (resp. $v_{t}^{X_j}$) denotes a prefix (resp. suffix) of $y_{t}^{X_j}$ of length $\Delta_{L} + \Delta_{R} + 4$ for any variable $X_j$
(see also Figure~\ref{fig:hat_z_construction}).
Also, we need $A_{t} u_{t}^{X_n}$ and $v_{t}^{X_n} B_{t}$ to create $\hat{A}_{t+1}^{X_n}$ and $\hat{B}_{t+1}^{X_n}$, respectively.

Note that by Definition~\ref{def:xshrink}, $|z_{t}^{X}| > \Delta_{L} + \Delta_{R} + 9$ if $z_{t}^{X} \neq \varepsilon$.
Then, we can compute $u_{t}^{X_i}$ for all variables $X_i$ in $O(n \log^*M)$ time and space
by processing variables in increasing order on the basis of the following fact:
$u_{t}^{X_i} = u_{t}^{X_{\ell}}$ if $z_{t}^{X_{\ell}} \neq \varepsilon$,
otherwise $u_{t}^{X_i}$ is the prefix of $z_{t}^{X_i}$ of length $\Delta_{L} + \Delta_{R} + 4$.
Similarly $v_{t}^{X_i}$ for all variables $X_i$ can be computed in $O(n \log^*M)$ time and space.

Using $u_{t}^{X_i}$ and $v_{t}^{X_i}$ for all variables $X_i$,
we can obtain $O(n \log^*M)$ blocks of signatures to which we give signatures.
We determine signatures by sorting the blocks by bucket sort as Lemma~\ref{lem:signature_encblock}
in $O(n \log^*M)$ time.

Hence, we can compute $\hat{\Lambda}_{t+1}$ in $O(n \log^*M)$ time and $O(n \log^*M)$ space.
\end{proof}

\begin{figure}[t]
\begin{center}
  \includegraphics[scale=0.5]{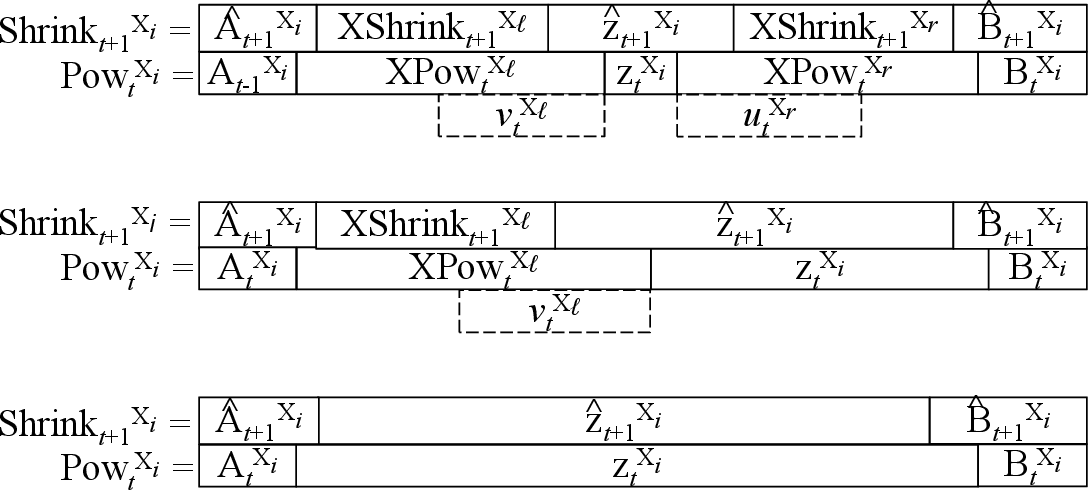}
  \caption{
  Abstract images of the needed signature sequence $v_{t}^{X_{\ell}} z_{t}^{X_i} u_{t}^{X_r}$ ($v_{t}^{X_{\ell}}$ and $u_{t}^{X_r}$ are not shown when they are empty)
  for computing $\hat{z}^{X_i}_{t+1}$ in three situations:
  Top for $0 \leq t < h^{X_{\ell}}, h^{X_{r}}$; middle for $h^{X_r} \leq t < h^{X_{\ell}}$; and bottom for $h^{X_{\ell}}, h^{X_{r}} \leq t < h^{X_i}$.
  } 
  \label{fig:hat_z_construction}
\end{center}
\end{figure}

We are ready to prove the theorem.
\begin{proof}
Using Lemmas~\ref{lem:Lambda_0},~\ref{lem:Lambda_t} and~\ref{lem:hat_Lambda_t},
we can get $\hat{\Lambda}_{h^{X_n}}$ in $O(n \log \log n \log N \log^*M)$ time
by computing $\Lambda_{0}, \hat{\Lambda}_{1}, \ldots, \hat{\Lambda}_{h^{X_n}}$ incrementally.
Note that during the computation we only have to keep $\Lambda_{t}$ (or $\hat{\Lambda}_{t}$) for the current $t$ and the assignments of $\mathcal{G}$.
Hence the working space is $O(n \log^* M + w)$.
By processing $\hat{\Lambda}_{h^{X_n}}$ in $O(n \log^* M)$ time,
we can get a sorted assignment set $\mathcal{D}$ of $\mathcal{G}$ of size $O(w)$.
Finally, we process $\mathcal{G}$ in $O(w)$ time and space to get $\mathcal{H}(\log w, w)$ by Lemma~\ref{lem:linear_construction}.
\end{proof}

\section{Appendix: Applications}\label{sec:appendix_applications}
\subsection{Proof of Theorem~\ref{theo:lzfac}}

For integers $j, k$ with $1 \leq j \leq j+k-1 \leq N$,
let $\mathit{Fst}(j,k)$ be the function which returns the minimum integer $i$ 
such that $i < j$ and $T[i..i+k-1] = T[j..j+k-1]$, if it exists. 
Our algorithm is based on the following fact:
\begin{fact}\label{fact:lz77}
Let $f_1, \ldots, f_z$ be the LZ77-Factorization of a string $T$. 
Given $f_1, \ldots, f_{i-1}$, we can compute $f_{i}$ with $O(\log |f_i|)$ calls of $\mathit{Fst}(j,k)$
(by doubling the value of $k$, followed by a binary search), 
where $j = |f_1 \cdots f_{i-1}|+1$.
\end{fact}

We explain how to support queries $\mathit{Fst}(j,k)$ using the signature encoding.
We define $e.{\rm min} = \min \mathit{vOcc}(e,S) + |e.{\rm left}|$ for a signature $e \in \mathcal{V}$ with
$e \rightarrow e_{\ell}e_{r}$ or $e \rightarrow \hat{e}^{k}$.
We also define $\mathit{FstOcc}(P,i)$ for a string $P$ and an integer $i$ as follows:
\[
  \mathit{FstOcc}(P,i) = \min \{ e.{\rm min} \mid (e, i) \in \mathit{pOcc}_{\mathcal{G}}(P,i) \}
\]
Then $\mathit{Fst}(j,k)$ can be represented by $\mathit{FstOcc}(P,i)$ as follows:
\begin{eqnarray*}
\mathit{Fst}(j,k) &=& \min \{ \mathit{FstOcc}(T[j..j+k-1],i) - i \mid i \in \{ 1,\ldots, k-1 \} \\
&=& \min \{ \mathit{FstOcc}(T[j..j+k-1],i) - i \mid i \in \mathcal{P} \},
\end{eqnarray*}
where $\mathcal{P}$ is the set of integers in Lemma~\ref{lem:pattern_occurrence_lemma1} with $P = T[j..j+k-1]$.

Recall that in Section~\ref{subsec:StaticIndexSigDic} 
we considered the two-dimensional orthogonal range reporting problem to enumerate $\mathit{pOcc}_{\mathcal{G}}(P,i)$.
Note that $\mathit{FstOcc}(P,i)$ can be obtained by taking $(e, i) \in \mathit{pOcc}_{\mathcal{G}}(P,i)$ with $e.{\rm min}$ minimum.
In order to compute $\mathit{FstOcc}(P,i)$ efficiently instead of enumerating all elements in $\mathit{pOcc}_{\mathcal{G}}(P,i)$,
we give every point corresponding to $e$ the weight $e.{\rm min}$ and
use the next data structure to compute a point with the minimum weight in a given rectangle.
\begin{lemma}[\cite{DBLP:journals/comgeo/AgarwalAG0Y13}]\label{lem:minimum_weight_report}
Consider $n$ weighted points on a two-dimensional plane.
There exists a data structure which supports the query to return a point with the minimum weight in a given rectangle 
in $O(\log^2 n)$ time, occupies $O(n)$ space, and requires $O(n \log n)$ time to construct.
\end{lemma}

Using Lemma~\ref{lem:minimum_weight_report}, we get the following lemma.
\begin{lemma}\label{lem:prevlemma}
Given a signature encoding $\mathcal{G} = (\Sigma, \mathcal{V}, \mathcal{D}, S)$ of size $w$ which generates $T$,
we can construct a data structure of $O(w)$ space in $O(w \log w \log N \log^* N)$ time
to support queries $\mathit{Fst}(j,k)$ in $O(\log w \log k \log^* N (\log N + \log k \log^* N))$ time.
\end{lemma}
\begin{proof}
For construction, we first compute $e.{\rm min}$ in $O(w)$ time using the DAG of $\mathcal{G}$.
Next, we prepare the plane defined by the two ordered sets $\mathcal{X}$ and $\mathcal{Y}$ in Section~\ref{subsec:StaticIndexSigDic}.
This can be done in $O(w \log w \log N \log^* N)$ time by sorting elements in $\mathcal{X}$ (and $\mathcal{Y}$)
by $\LCEQ$ algorithm (Lemma~\ref{lem:sub_operation_lemma}) and a standard comparison-based sorting.
Finally we build the data structure of Lemma~\ref{lem:minimum_weight_report} in $O(w \log w)$ time.

To support a query $\mathit{Fst}(j,k)$, we first compute $\encpow{\uniq{P}}$
with $P = T[j..j+k-1]$ in $O(\log N + \log k \log^* N)$ time by Lemma~\ref{lem:ComputeShortCommonSequence},
and then get $\mathcal{P}$ in Lemma~\ref{lem:pattern_occurrence_lemma1}.
Since $|\mathcal{P}| = O(\log k \log^* M)$ by Lemma~\ref{lem:common_sequence2},
$\mathit{Fst}(j,k) = \min \{ \mathit{FstOcc}(P,i) - i \mid i \in \mathcal{P} \}$
can be computed by answering $\mathit{FstOcc}$ $O(\log k \log^* M)$ times.
For each computation of $\mathit{FstOcc}(P,i)$,
we spend $O(\log w (\log N + \log k \log^* N))$ time to compute $x_1^{(P,j)}, x_2^{(P,j)}, y_1^{(P,j)}$ and $y_2^{(P,j)}$ by Lemma~\ref{lem:ComputePatternRange},
and $O(\log^2 w)$ time to compute a point with the minimum weight in the rectangle $(x_1^{(P,j)}, x_2^{(P,j)}, y_1^{(P,j)}, y_2^{(P,j)})$.
Hence it takes $O(\log k \log^* M (\log w (\log N + \log k \log^* N) + \log^2 w)) = O(\log w \log k \log^* N (\log N + \log k \log^* N))$ time in total.
\end{proof}

We are ready to prove Theorem~\ref{theo:lzfac} holds.
\begin{proof}[Proof of Theorem~\ref{theo:lzfac}]
We first compute the signature encoding of $T$ in $O(|T| f_{\mathcal{A}})$ time 
and $O(f'_{\mathcal{A}} + w)$ working space by the algorithm of Theorem~\ref{theo:HConstructuionTheorem}~(1).
Using a data structure $\mathcal{H}(f_{\mathcal{A}},f'_{\mathcal{A}})$ achieving
$f_{\mathcal{A}} = O\left(\min \left\{ \frac{\log \log M \log \log w}{\log \log \log M}, \sqrt{\frac{\log w}{\log\log w}} \right\} \right)$ time 
and $f'_{\mathcal{A}} = O(w)$ space, the working space becomes $O(w)$ space.
Next we compute the $z$ factors of the LZ77-Factorization of $T$ incrementally 
by using Fact~\ref{fact:lz77} and Lemma~\ref{lem:prevlemma} in $O(z \log w \log^3 N (\log^* N)^2)$ time.
Therefore the statement holds.
\end{proof}

\subsection{Proof of Theorem~\ref{theo:changedSLP}}

\subsubsection{Proof of Theorem~\ref{theo:changedSLP}~(1)}
\begin{proof}
For any signature $e \in \mathcal{V}$ such that $e \rightarrow e_{\ell}e_r$, 
we can easily translate $e$ to a production of SLPs because the assignment is a pair of signatures, 
like the right-hand side of the production rules of SLPs.
For any signature $e \in \mathcal{V}$ such that $e \rightarrow \hat{e}^k$, 
we can translate $e$ to at most $2 \log k$ production rules of SLPs:
We create $t = \lfloor \log k \rfloor$ variables which represent
$\hat{e}^{2^1}, \hat{e}^{2^2}, \ldots, \hat{e}^{2^t}$ and concatenating them
according to the binary representation of $k$ to make up $k$ $\hat{e}$'s.
Therefore we can compute $\mathcal{S}$ in $O(w \log |T|)$ time.
\end{proof}

\subsubsection{Proof of Theorem~\ref{theo:changedSLP}~(2)}
\begin{proof}
Note that the number of created or removed signatures in $\mathcal{V}$ 
is bounded by $O(y + \log |T'| \log^* M)$ by Lemma~\ref{lem:concatenate1}. 
For each of the removed signatures, we remove the corresponding production from $\mathcal{S}$. 
For each of created signatures, we create the corresponding production and add it to $\mathcal{S}$ as in the proof of (1). 
Therefore Theorem~\ref{theo:changedSLP} holds. 
\qed
\end{proof}

\subsection{Proof of Theorem~\ref{theo:faster_LCP}}

We use the following known result.
\begin{lemma}[\cite{LongAlstrup}]\label{lem:lcp_lcs_on_H}
Using the DAG for a signature encoding $\mathcal{G} = (\Sigma, \mathcal{V}, \mathcal{D}, S)$,
we can support 
\begin{itemize}
 \item $\mathit{LCP}(s_1, s_2)$ in $O(\log |s_1| + \log |s_2|)$ time,
 \item $\mathit{LCS}(s_1, s_2)$ in $O((\log |s_1| + \log |s_2|) \log^* M)$ time
\end{itemize}
where $\id{s_1}, \id{s_2} \in \mathcal{V}$.
\end{lemma}
\begin{proof}

We compute $\mathit{LCP}(s_1, s_2)$ by $\mathit{LCE}(s_1,s_2,1,1)$,  
namely, we use the algorithm of Lemma~\ref{lem:sub_operation_lemma}. 
Let $P$ denote the longest common prefix of $s_1$ and $s_2$.
We use the notation $\hat{A}^{P}$ defined in Section~\ref{sec:HConstruction3-2}.
There exists a signature sequence $v = \hat{A}^{P}_{h^{P}}\xshrink{h^{P}}{P}R_{h^{P}-1}^{P}\hat{R}_{h^{P}-1}^{P} \cdots R_{0}^{P}\hat{R}_{0}^{P}$ 
that occurs at position $1$ in $\id{s_1}$ and $\id{s_2}$ by a similar argument of Lemma~\ref{lem:common_sequence2}.
Since $|\encpow{v}| = O(\log |P| + \log^* M)$, 
we can compute $\mathit{LCP}(s_1, s_2)$ in $O(\log |s_1| + \log |s_2|)$ time. 
Similarly, we can compute $\mathit{LCS}(s_1, s_2)$ in $O((\log |s_1| + \log |s_2|) \log^* M)$ time.
More detailed proofs can be found at~\cite{LongAlstrup}.
\qed
\end{proof}

To use Lemma~\ref{lem:lcp_lcs_on_H} for $\id{\val{X_1}}, \ldots, \id{\val{X_n}}$, we show the following lemma. 
\begin{lemma}\label{lem:SLPSignatureEncoding}
Given an SLP $\mathcal{S}$, 
we can compute $\id{\val{X_1}}, \ldots, \id{\val{X_n}}$ in \\
$O(n \log\log n \log N \log^* N)$ time and $O(n \log N \log^* N)$ space. 
\end{lemma}
\begin{proof}
Recall that the algorithm of Theorem~\ref{theo:HConstructuionTheorem}~(3) 
computes $\id{\val{X_n}}$ in $O(n \log\log n \log N \log^* N)$ time.
We can modify the algorithm to compute $\id{\val{X_1}}, \ldots, \id{\val{X_n}}$ without changing the time complexity:
We just compute $A_{t}^{X}$, $\hat{A}_{t}^{X}$, $B_{t}^{X}$ and $\hat{B}_{t}^{X}$ for ``all'' $X \in \mathcal{S}$, not only for $X_n$.
Since the total size is $O(n \log N \log^* N)$, Lemma~\ref{lem:SLPSignatureEncoding} holds.
\end{proof}

We are ready to prove Theorem~\ref{theo:faster_LCP}.
\begin{proof}
The first result immediately follows from Lemma~\ref{lem:lcp_lcs_on_H} and~\ref{lem:SLPSignatureEncoding}. 
To speed-up query times for $\LCPQ$ and $\LCSQ$,
We sort variables in lexicographical order 
in $O(n \log n \log N)$ time by $\LCPQ$ query and a standard comparison-based sorting.
Constant-time $\LCPQ$ queries are then possible by 
using a constant-time RMQ data structure~\citex{DBLP:journals/jal/BenderFPSS05}
on the sequence of the lcp values.
$\LCSQ$ queries can be supported similarly. \qed
\end{proof}

\subsection{Proof of Theorem~\ref{theo:smaller_LCE}}
\begin{proof}
We can compute $\mathcal{H}(\log w, w)$
for a signature encoding $\mathcal{G} = (\Sigma, \mathcal{V}, \mathcal{D}, S)$ of size $w$ 
representing $T$ in $O(n \log \log n \log N \log^* N)$ time 
and $O(n \log^* M + w)$ working space using Theorem~\ref{theo:HConstructuionTheorem}, 
where $w = O(z \log N \log^* N)$.
Notice that 
each variable of the SLP appears at least once in the 
derivation tree of $T_{n}$ of the last variable $X_n$ representing the string $T$.
Hence, if we store an occurrence of each variable $X_i$ in $\mathcal{T}_n$
and $|\val{X_i}|$, we can reduce any LCE query on two variables 
to an LCE query on two positions of $\val{X_n} = T$.
In so doing, for all $1 \leq i \leq n$  
we first compute $|\val{X_i}|$ 
and then compute the leftmost occurrence $\ell_i$ of $X_i$ in $\mathcal{T}_n$,
spending $O(n)$ total time and space.
By Lemma~\ref{lem:sub_operation_lemma}, 
each LCE query can be supported in $O(\log N + \log \ell \log^* N)$ time.
Since $z \leq n$~\citex{rytter03:_applic_lempel_ziv},
the total preprocessing time is $O(n \log \log n \log N \log^* N)$
and working space is $O(n \log^* M + w)$.
\qed
\end{proof}
\subsection{Proof of Theorem~\ref{theo:palindrome}}

\begin{proof}
For a given SLP of size $n$ representing a string of length $N$,
let $P(n,N)$, $S(n, N)$, and $E(n,N)$ be
the preprocessing time and space requirement for an $\LCEQ$ data structure on
SLP variables, and each $\LCEQ$ query time, respectively.

Matsubara et al.~\cite{matsubara_tcs2009} showed that 
we can compute an $O(n \log N)$-size representation
of all palindromes in the string
in $O(P(n,N) + E(n,N) \cdot n \log N)$ time and $O(n \log N + S(n, N))$ space.
Hence, using Theorem~\ref{theo:smaller_LCE},
we can find all palindromes in the string in 
$O(n \log \log n \log N \log^* N + n \log^2 N \log^* N) = O(n \log^2 N \log^* N)$ time
and $O(n \log^* N + w)$ space.
\qed
\end{proof}

\subsection{Proof of Theorem~\ref{theo:Lyndon}}

\begin{proof}
It is shown in~\cite{INIBT13} that 
we can compute the Lyndon factorization of the string 
in $O(P(n,N) + E(n,N) \cdot n \log n)$ time using $O(n^2+S(n,N))$ space.
Hence, using Theorem~\ref{theo:smaller_LCE},
we can compute the Lyndon factorization of the string in 
$O(n \log \log n \log N \log^* N + n \log n \log N \log^* N) = O(n \log n \log N \log^* N)$ time.
We remark that since $m \leq n$ due to~\cite{INIBT13},
the output size $m$ is omitted in the total time complexity.
\qed
\end{proof}

\subsection{Proof of Theorem \ref{theo:ImproveDictionaryMatching}}

\begin{proof}
In the preprocessing phase, 
we construct an $\mathcal{H}(\log w', w')$
for a signature encoding $\mathcal{G} = (\Sigma, \mathcal{V}, \mathcal{D}, S)$ of size $w'$ 
such that \\
$\id{\val{X_{1}}}, \ldots ,\id{\val{X_n}} \in \mathcal{V}$ 
using Lemma~\ref{lem:SLPSignatureEncoding}, spending
$O(n \log \log n \log N \log^* M)$ time, where $w' = O(n \log N \log^* M)$. 
Next we construct a compacted trie of size $O(m)$ that represents the $m$ patterns, which we denote by \emph{$\mathit{PTree}$ (pattern tree)}.
Formally, each non-root node of $\mathit{PTree}$ represents either a pattern or 
the longest common prefix of some pair of patterns.
$\mathit{PTree}$ can be built by using $\LCPQ$ of Theorem~\ref{theo:faster_LCP} in $O(m \log m \log N)$ time.
We let each node have its string depth, and the pointer to its deepest ancestor node that represents a pattern if such exists.
Further, we augment $\mathit{PTree}$ with a data structure for level ancestor queries so that 
we can locate any prefix of any pattern, designated by a pattern and length, in $\mathit{PTree}$ in $O(\log m)$ time
by locating the string depth by binary search on the path from the root to the node representing the pattern.
Supposing that we know the longest prefix of $T[i..|T|]$ that is also a prefix of one of the patterns, which we call the \emph{max-prefix for $i$},
$\mathit{PTree}$ allows us to output $\mathit{occ}_i$ patterns occurring at position $i$ in $O(\log m + \mathit{occ}_i)$ time.
Hence, the pattern matching problem reduces to computing the max-prefix for every position.

In the pattern matching phase, our algorithm processes $T$ in a streaming fashion, i.e.,
each character is processed in increasing order and discarded before taking the next character.
Just before processing $T[j+1]$, 
the algorithm maintains a pair of signature $p$ and integer $l$
such that $\derive(p)[1..l]$ is the longest suffix of $T[1..j]$ that is also a prefix of one of the patterns.
When $T[j+1]$ comes, we search for the smallest position $i \in \{j-l+1, \dots, j+1\}$ such that there is a pattern whose prefix is $T[i..j+1]$.
For each $i \in \{j-l+1, \dots, j+1\}$ in increasing order,
we check if there exists a pattern whose prefix is $T[i..j+1]$ by binary search on a sorted list of $m$ patterns.
Since $T[i..j] = \derive(p)[i-j+l..l]$, $\LCEQ$ with $p$ can be used for comparing a pattern prefix and $T[i..j+1]$ (except for the last character $T[j+1]$),
and hence, the binary search is conducted in $O(\log m \log N \log^* M)$ time.
For each $i$, if there is no pattern whose prefix is $T[i..j+1]$,
we actually have computed the max-prefix for $i$, and then we output the occurrences of patterns at $i$.
The time complexity is dominated by the binary search, which takes place $O(|T|)$ times in total.
Therefore, the algorithm runs in $O(|T|\log m \log N \log^* N + \mathit{occ})$ time.

By the way, one might want to know occurrences of patterns as soon as they appear
as Aho-Corasick automata do it by reporting the occurrences of the patterns by their ending positions.
Our algorithm described above can be modified to support it without changing the time and space complexities.
In the preprocessing phase, we additionally compute \emph{$\mathit{RPTree}$ (reversed pattern tree)},
which is analogue to $\mathit{PTree}$ but defined on the reversed strings of the patterns,
i.e., $\mathit{RPTree}$ is the compacted trie of size $O(m)$ that represents the reversed strings of the $m$ patterns.
Let $T[i..j]$ be the longest suffix of $T[1..j]$ that is also a prefix of one of the patterns.
A suffix $T[i'..j]$ of $T[i..j]$ is called the \emph{max-suffix for $j$} iff
it is the longest suffix of $T[i..j]$ that is also a suffix of one of the patterns.
Supposing that we know the max-suffix for $j$,
$\mathit{RPTree}$ allows us to output $\mathit{eocc}_j$ patterns occurring with ending position $j$ in $O(\log m + \mathit{eocc}_j)$ time.
Given a pair of signature $p$ and integer $l$ such that $T[i..j] = \derive(p)[1..l]$,
the max-suffix for $j$ can be computed in $O(\log m \log N \log^* N)$ time by binary search on a list of $m$ patterns sorted by their ``reversed'' strings
since each comparison can be done by ``leftward'' $\LCEQ$ with $p$.
Except that we compute the max-suffix for every position and output the patterns ending at each position,
everything else is the same as the previous algorithm, and hence, the time and space complexities are not changed.
\end{proof}

\clearpage
\bibliographystylex{splncs03}
\bibliographyx{ref}
\end{document}